\documentclass[12pt,english,aps,prl,twocolumn,showpacs,superscriptaddress,
footinbib,reprint,noshowpacs]{revtex4-1}
\usepackage{graphicx}
\def\arXiv#1{\href{http://arxiv.org/abs/#1}{arXiv:#1}}
\usepackage{amsmath}
\usepackage{comment}
\usepackage{amssymb}
\usepackage{caption}
\usepackage{soul}

\RequirePackage{amsmath,amssymb,amsthm,graphicx,mathrsfs,url}
\RequirePackage[usenames,dvipsnames]{color}
\RequirePackage[colorlinks=true,linkcolor=Red,citecolor=Green]{hyperref}
\RequirePackage{amsxtra}

\newtheorem{theo}{Theorem}
\newtheorem{prop}{Proposition}

\newtheorem{lemm}[prop]{Lemma}
\newtheorem{corr}[prop]{Corollary}

\usepackage{scalerel}

\DeclareMathOperator{\Spec}{Spec}

\let\Im=\Imag

\let\Re=\Real

\newcommand{\RR}{{\mathbb R}}

\newcommand{\CC}{{\mathbb C}}

\newcommand{\ZZ}{{\mathbb Z}}
\usepackage{wasysym}

\usepackage{subcaption} 
\def\smallsection#1{\smallskip\noindent\textbf{#1}.}
\begin{document}

\title{Spectral characterization of magic angles in twisted bilayer graphene}

\author{Simon Becker}
\email{simon.becker@damtp.cam.ac.uk}
\address{Department of Applied Mathematics and Theoretical Physics, University of Cambridge, Wilberforce Road, Cambridge, CB3 0WA, United Kingdom.}

\author{Mark Embree}
\email{embree@vt.edu}
\address{Department of Mathematics, Virginia Tech, 
Blacksburg, VA 24061, USA}

\author{Jens Wittsten}
\email{jens.wittsten@math.lu.se}
\address{Centre for Mathematical Sciences, Lund University, Box 118, SE-221 00 Lund, Sweden, and Department of Engineering, University of Bor{\aa}s, SE-501 90 Bor{\aa}s, Sweden}

\author{Maciej Zworski}
\email{zworski@math.berkeley.edu}
\address{Department of Mathematics, University of California,
Berkeley, CA 94720, USA.}

\begin{abstract}
Twisted bilayer graphene (TBG) has been experimentally observed to exhibit almost flat bands when the twisting occurs at certain {\em magic angles}. In this letter %, we report new results on the continuum model of twisted bilayer graphene and its electronic band structure. Under 
we show that in the approximation of vanishing AA-coupling,
 the magic angles (at which there exist entirely flat bands) are given as the eigenvalues of a non-hermitian operator, and that all bands start squeezing exponentially fast as the  angle $\theta$ tends to $0$. In particular, as the interaction potential changes,
the dynamics of magic angles involves the non-physical complex
 eigenvalues. Using our new spectral characterization, we show that the equidistant scaling of inverse magic angles, as observed in \cite{TKV}, is special for the choice of tunnelling potentials in the continuum model, and is not protected by symmetries. While we also show that the protection of zero-energy states holds in the continuum model as long as particle-hole symmetry is preserved, we observe that the existence of flat bands and the exponential squeezing are special properties of the chiral model. 
\end{abstract} 

\maketitle
\makeatother

The electronic structure in a stack of two graphene sheets depends crucially on the relative twist angle $\theta$ of the layers \cite{M10,LMM10,SSKP10,MCBPB10,MK12,TMM12,MK13,WSFRSP16}. Tunneling interactions coming from long-period Moir\'e patterns (emerging from the twisting) cause the Dirac dispersion relation of the non-interacting sheets to change dramatically: the Floquet bands become nearly flat at a certain discrete set of \emph{magic angles}, as observed in \cite{BM,BM2,LPC}. This flattening of bands led to the discovery of Mott insulation \cite{Cao2} and of unconventional superconductivity in TBG \cite{Cao1,Yank3,S20,PZVS}.
 
 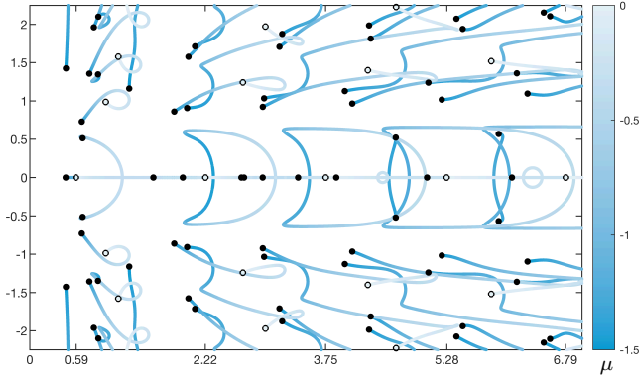
\begin{figure}[h!] \begin{center}
\includegraphics[width=8.5cm]{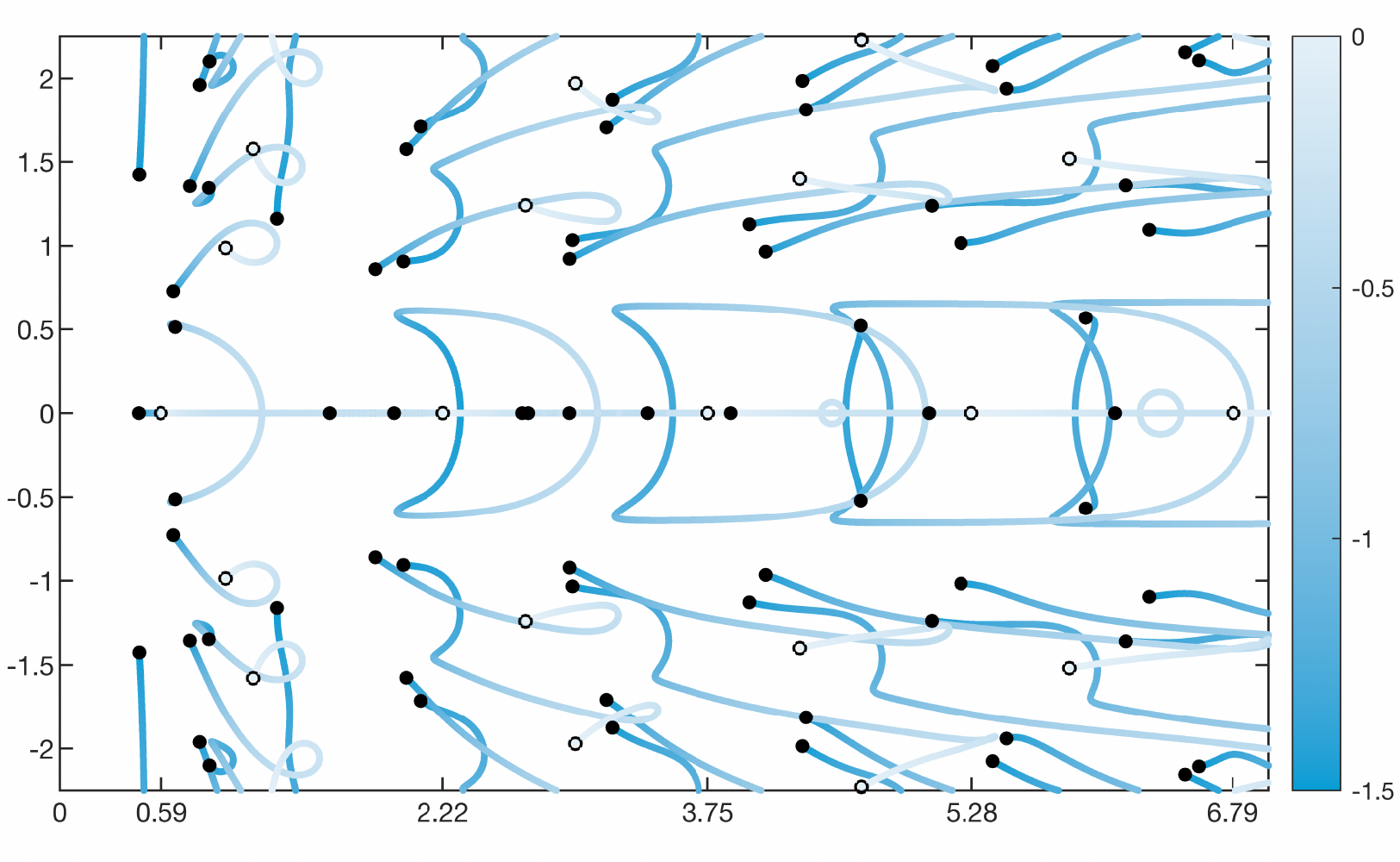}  
\begin{picture}(0,0)
\put(-19,4.5){\footnotesize $\mu$}
\end{picture}
\end{center}

\vspace*{-17pt}
\caption{For a tunnelling potential $U_{\mu}$ in \eqref{eq:tunnell} we find non-equidistant spacing of the inverse magic angles: 
$\bullet$~denote $w_1$'s for $ \mu = -1.5$; the paths trace the 
 $w_1$ dynamics for $ -1.5 \leq \mu \leq 0$; $\circ$ denote the
$w_1$'s for $\mu=0$.}\label{fig:equidistant}
\end{figure}
 
We give a new spectral characterization of magic angles for the reduced model without AA-coupling (\emph{chiral model} of \cite{TKV}) as eigenvalues of a compact non-hermitian Birman-Schwinger operator. Since its spectrum has also complex eigenvalues, we establish that flat bands also exist for complex parameters in the continuum model. 
The spectral characterization also allows an efficient numerical calculation of a large set of magic angles, as demonstrated in \cite[\S 5]{BEWZ}. Consequently, we are able to study a larger number of magic angles and find that the equidistant spacing of reciprocal angles observed in \cite{TKV} is unique to the particular tunnelling potential and not protected by symmetries -- see Fig.~\ref{fig:equidistant} and \ref{f:resa}. In addition, we show that in the chiral model the lowest bands become exponentially small as the inverse twisting angle $\theta^{-1}$ tends to infinity; see also \cite{Xu} for an experimental study of small twist angles.

\begin{figure}
\begin{center}
\includegraphics[width=8.5cm]{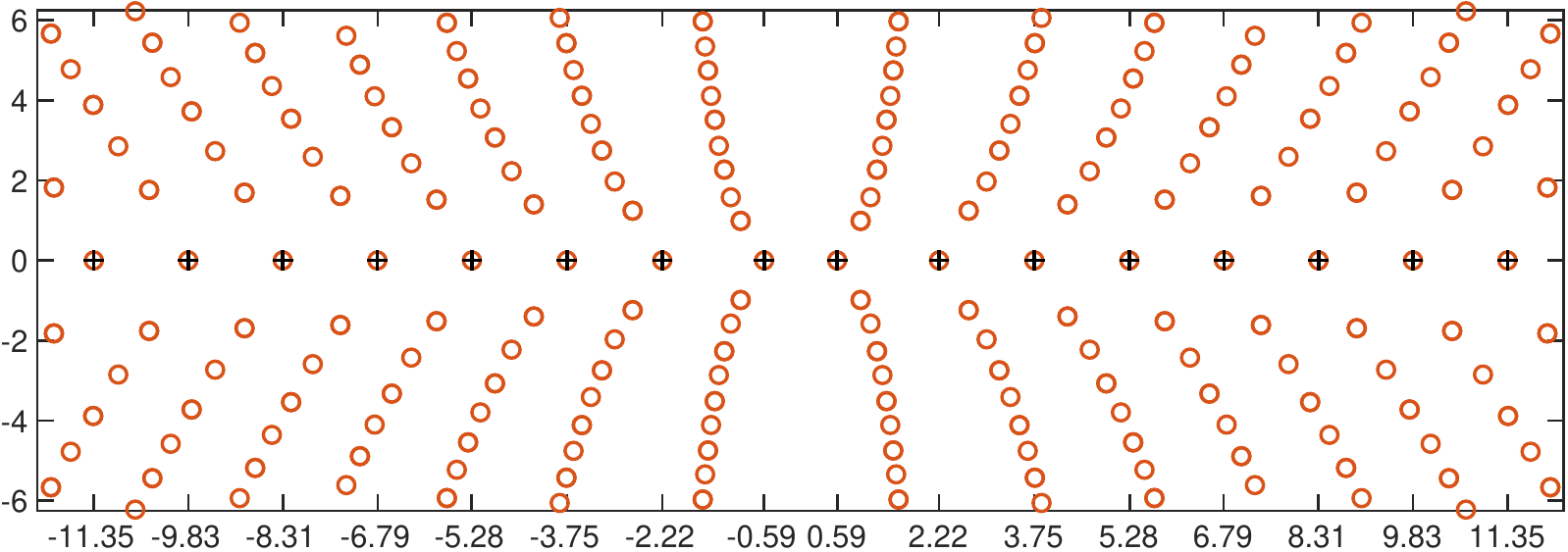}
\end{center}

\vspace*{-12pt}
\caption{\label{f:resa}
Reciprocals of magic angles for potential~\eqref{eq:defU} 
from~\cite{TKV} in the complex plane: 
resonant $ w_1$'s  ($\circ$) come from the full spectrum
of the operator \eqref{eq:BSop} defining magic angles,
and the (real-valued) magic $ w_1$'s ($+$) are the reciprocals
of the \emph{physically relevant} positive angles.}
\end{figure}

We then analyze the particle-hole symmetric continuum model with additional AA-coupling. Based on a numerical analysis of this model and a rigorous argument for the extremal model with only AA-coupling (\emph{anti-chiral model} of \cite{TKV}), we find that, unlike in the chiral model, there are no entirely flat bands at zero energy, and bands are also no longer squeezed with an exponential rate. However, we show that the zero energy level of the model with particle-hole symmetry, is still protected in the spectrum, which has been previously observed in \cite{TKV} for the chiral model.

%By setting the AA-coupling to zero, \cite{TKV} derived a chiral Hamiltonian that features completely flat bands. Thus, they extended the notion of \emph{magic angles} from -- as previously defined -- angles at which the Fermi velocity, $v_F:=\partial_k E_k$, at the Dirac points vanishes to angles at which the entire zeroth bands becomes flat. 

The Hamiltonian of the continuum model for twisted bilayer graphene is defined for variables 
\[ z=x_1+ix_2, \ \  D_{\bar z} = \tfrac{1}{2i} (\partial_{x_1}+i \partial_{x_2}) , \ \  \omega=e^{2\pi i/3}, \]
coupling parameters $w=(w_0,w_1)$ and twisting angles $\varphi$ of the Moir\'e Dirac cones, which we allow to vary independently from the mechanical twisting angle $\theta$, by
\begin{equation}
\label{eq:Hamiltonian}
 H ( w, \varphi ) := \begin{pmatrix} 
C ( w_0 ) & D( w_1 , \varphi )^*  \\
D( w_1, \varphi )  & C ( w_0 ) \end{pmatrix}, 
\end{equation}
where
\begin{equation}
\begin{split}
\label{eq:D}
D(w_1, \varphi )  &:= \begin{pmatrix} 2 e^{i\varphi/2} D_{\bar z} & w_1 U(z) \\ w_1 U(-z) & 2e^{-i\varphi/2}D_{\bar z} \end{pmatrix} , \\ C(w_0 )&:=\begin{pmatrix}0 & w_0 V( z )  \\  w_0 V ( - z ) & 0 \end{pmatrix}. 
\end{split}
\end{equation}
Here, the tunnelling potentials \cite{BM,TKV} are given by
\begin{equation} 
\label{eq:defU} U(z) := \sum_{k=0}^2 \omega^k e^{\frac{1}{2}(z\bar \omega^k- \bar z \omega^k)}, \text{ and }
 V(z) := 2 \partial_z U ( z ). 
 \end{equation}
The parameter $w_0$ controls the AA-coupling, whereas $w_1$ controls AB- and BA-coupling. By switching off the AA-coupling, $w_0=0$, it has been noticed in \cite{TKV} that the spectrum of the Hamiltonian becomes independent of the the prefactors $\theta$ and the parameter $w_1^{-1}$ is proportional to $\theta$. The potential $U$ satisfies three key symmetry properties:
\begin{subequations}\label{eq:symmU}
\vspace*{-12pt}
\begin{equation}
 U ( z + \tfrac 43 \pi i \omega^\ell ) = \bar \omega U ( z ), \qquad \mbox{for $\ell=1,2$}, 
\end{equation}
\vspace*{-22pt}
\begin{equation}
 U ( \omega z ) = \omega U ( z ) ,\qquad  \overline{ U (\bar z ) } = U ( z ),
 \end{equation}
\end{subequations}
and most of the results here apply to more general potentials satisfying \eqref{eq:symmU}, for instance \eqref{eq:tunnell}. Such potentials are discussed in the supplementary material, see \eqref{eq:gen_pot}.
  
It is immediate from \eqref{eq:symmU} that the Hamiltonian is periodic with respect to the lattice $\Gamma:=4\pi (i \omega \mathbb{Z} \oplus i \omega^2 \mathbb{Z})$, and by Floquet theory magic angles are defined as angles $\theta= w_1^{-1}$ such that the energy $0$ is in the spectrum 
(on $ \Gamma$-periodic functions) of all Hamiltonians $H_{\mathbf k}(w,\varphi)$ where 
\begin{equation}\label{eq:Hk}
H_{\mathbf k} ( w, \varphi ) := \begin{pmatrix} 
C(w_0) & D( w_1 , \varphi )^*-\overline{\mathbf k}  \\
D( w_1, \varphi )- \mathbf k  & C(w_0) \end{pmatrix}
\end{equation}
where $\mathbf k \in \mathbb C$ denotes the quasi-momentum.

The Hamiltonian with only AB/BA-coupling, $w_0=\varphi=0,$ is called the \emph{chiral} continuum model, the opposite choice of only $AA$-coupling, $w_1=\varphi=0$, is called the \emph{anti-chiral} continuum model, and the choice $\varphi=0$ the \emph{particle-hole symmetric} continuum model.

The Mott insulation, experimentally observed in TBG \cite{Cao2}, is due to strongly correlated electron interactions~\cite{MP37}; the exponential squeezing of bands described here provides a qualitative explanation for said effect, as kinetic energy in squeezed bands and electronic interactions dominate \cite{BJ}. In addition, it has been observed \cite{S20,Yank3} that the correlated insulating and superconducting regimes are not limited to magic angles but persist also at angles close to the magic ones, as long as bands remain narrow. This suggests that the physically relevant phenomenon is the squeezing and flattening of bands, rather than the existence of an entirely flat band (which is seemingly unstable under perturbations).

\smallsection{Spectral characterization of magic angles} To obtain a spectral characterization of magic angles for the chiral Hamiltonian, we define the \emph{Birman-Schwinger operator} for $\mathbf k \notin \Gamma^*$ (the dual lattice) by
\begin{equation}
\label{eq:BSop}
T_{\mathbf k } := ( 2 D_{\bar z } - \mathbf k )^{-1} \begin{pmatrix}
0 & U ( z ) \\ U ( - z ) & 0 \end{pmatrix}.
\end{equation}
For zero AA-coupling, there is a discrete set of values $w_1 \in \mathcal A := 1/  \Spec ( T_{\mathbf k } ) $, where $\Spec$ denotes the spectrum
(always on $ \Gamma$-periodic functions), that is, quite remarkably, independent of $\mathbf k \notin \Gamma^*$ \cite[Theorem 2]{BEWZ}. In addition, the set $\mathcal A$ satisfies symmetries $\mathcal A = - \mathcal A = \overline {\mathcal A }$ \cite[Proposition 3.2]{BEWZ}; see Figs.~\ref{fig:equidistant} and \ref{f:resa}. 

The flat bands at zero of the chiral Hamiltonian then occur precisely at $w_1 \in \mathcal A$. The set of magic angles is then just given by angles $\theta$ for which $\theta^{-1} \in \mathcal A \cap \RR,$ see \cite[Theorem~2]{BEWZ} for details.

In addition, we can characterize the magic angles for the chiral model by analyzing the spectrum of the non-hermitian operator $D(w_1)=D(w_1,0)$ in \eqref{eq:D}. In fact, $\theta=w_1^{-1}$ is a magic angle if and only if $\Spec(D(w_1))=\CC.$ For all other $\theta$, $\Spec(D(w_1))=\Gamma^*$, the dual lattice~\cite[Theorem~2]{BEWZ}. Fig.~\ref{f:psa} illustrates this remarkable discontinuity in the spectrum of $D(w_1)$.
\begin{figure}
\hspace*{-10pt}
\includegraphics[width=4.65cm]{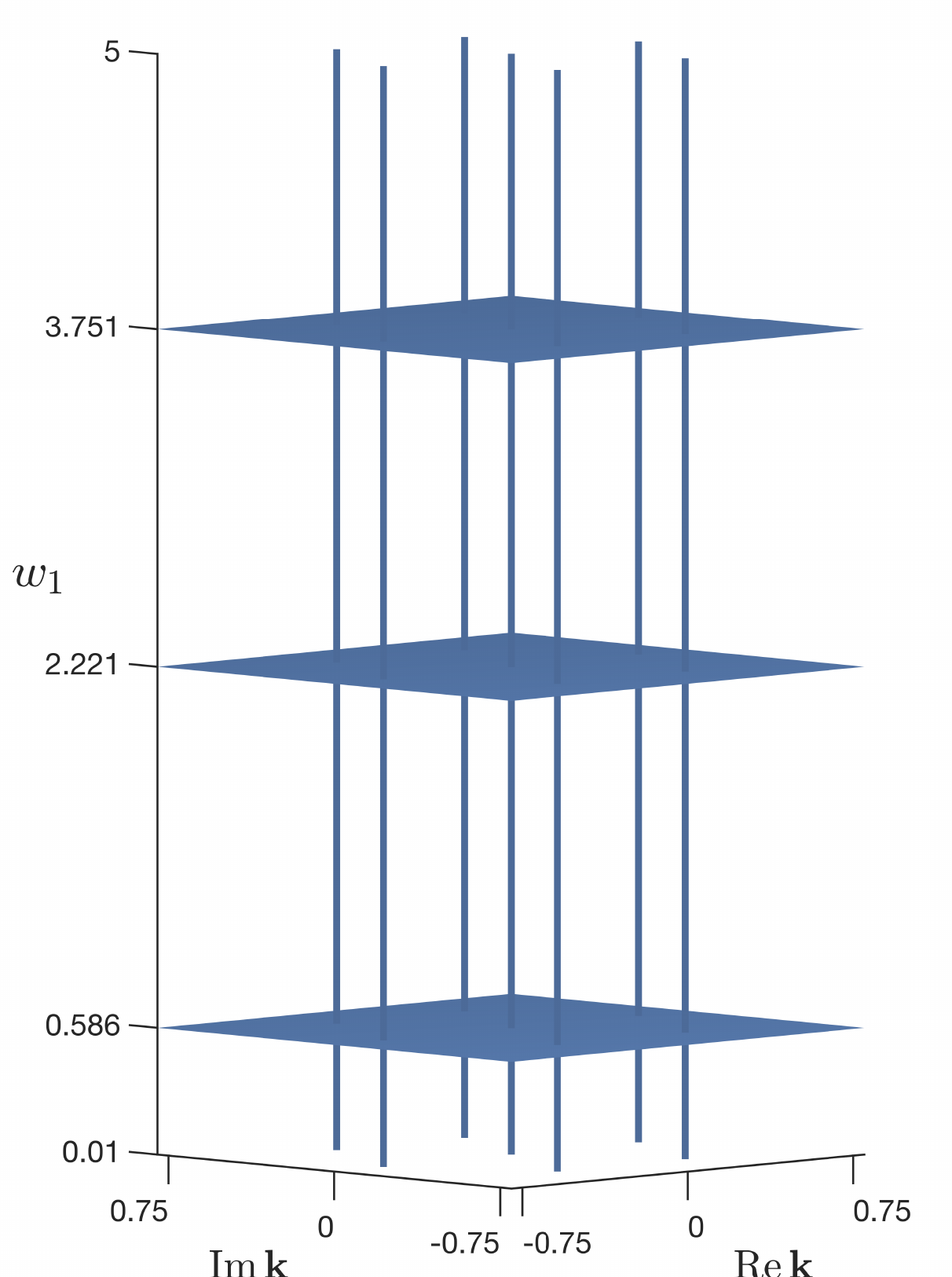}
\hspace*{-25pt}
\includegraphics[width=4.65cm]{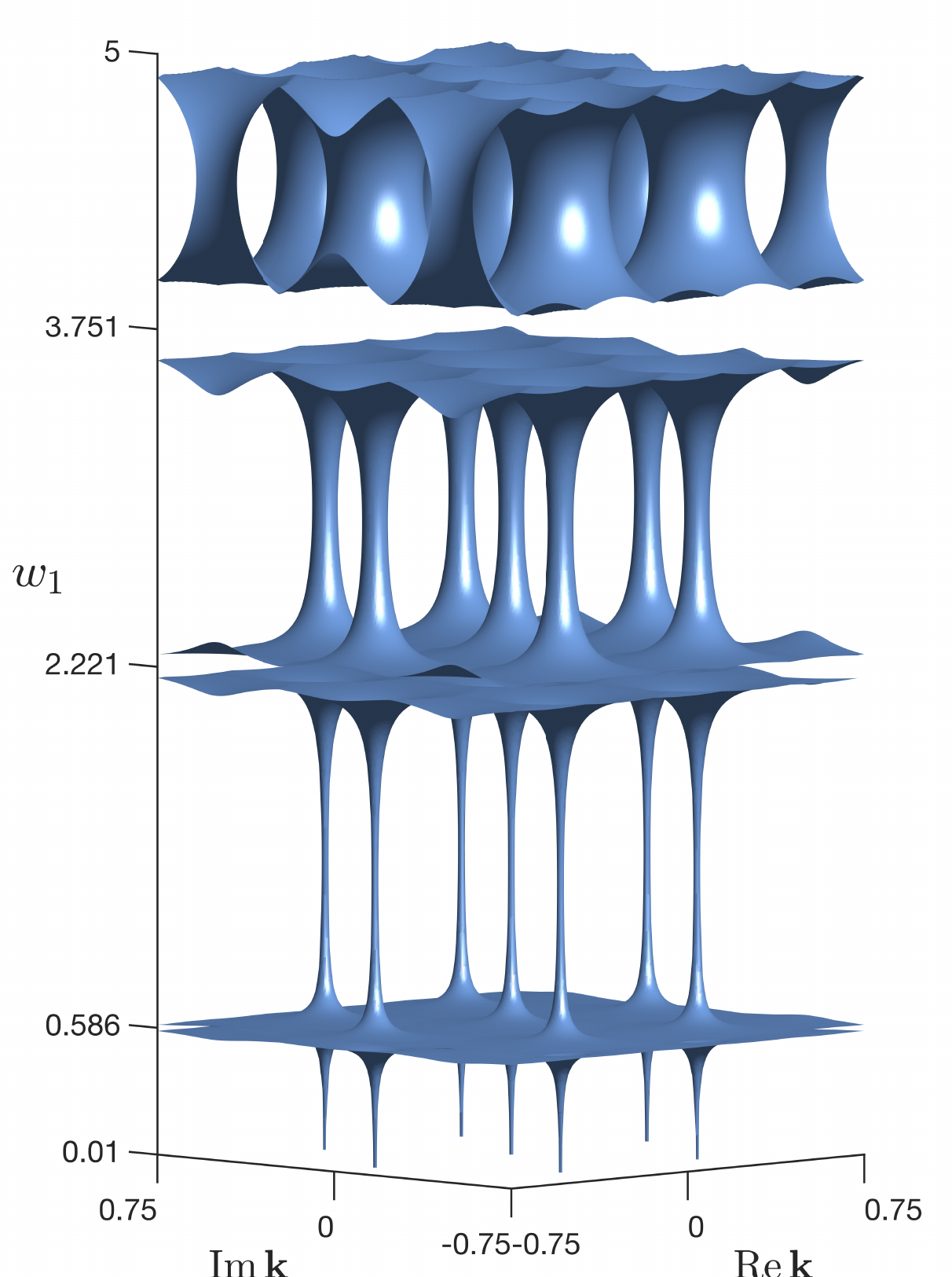} 
\caption{\label{f:psa}: Left panel: the spectrum of $ D (w_1 ) $ 
(in the $ \mathbf k $ plane) as $w_1$ varies (vertical axis). Flat surfaces indicate that $1/w_1$ is a magic angle. 
Right panel: level surface of $ \| ( D ( w_1 ) - \mathbf k )^{-1} \| = 10^{\frac 94}$ as a function of
$ \mathbf k $ and
 $w_1$:
 %also 
 %$  \| ( D ( w_1 ) - \mathbf k )^{-1} \| \geq 
 %e^{ w_1 } $ (see \eqref{eq:sque} and the discussion that follows) 
the norm blows up at magic angles for all $ {\bf k} $ ($w_1$ near $ 0.586$,  $2.221$, and $3.751$).  
The thickening of the ``trunks'' reflects the exponential squeezing
of the smallest eigenvalue of $H_{\bf k}((0,w_1),0)$ as $w_1$ grows,
relating to {\em pseudospectral} effects \cite{dsz,TE} -- see \eqref{eq:sque} and the discussion that follows.}
%In any discretization of $D(w_1)$ the norm of the resolvent would be finite 
%at magic $1/w_1$, except on a finite set, but this norm would blow up as the discretization improves, 
\end{figure}

\smallsection{Asymptotic distribution of magic angles} It has been observed in \cite{TKV} that the differences of reciprocal magic angles, with $0 <\theta_k <\theta_{k-1}$, for the chiral model behave asymptotically like $\theta_k^{-1}-\theta_{k-1}^{-1} \simeq 3/2.$
Our new characterization of magic angles as real eigenvalues of the operator $T_{\mathbf k }$ can give magic angles to high precision: we have high confidence of all the digits shown:
\begin{center}
\begin{tabular}{rclcc}
\multicolumn{1}{c}{$k$} & &
\multicolumn{1}{c}{$\theta_k^{-1}$} &
 & $\theta_{k}^{-1}-\theta_{k-1}^{-1} $ \\[2pt] \hline
1  &\ &   \phantom{0}0.58566355838955 &\ &               \\
2  &&   \phantom{0}2.2211821738201  && 1.6355        \\
3  &&   \phantom{0}3.7514055099052  && 1.5302        \\
4  &&   \phantom{0}5.276497782985   && 1.5251        \\
5  &&   \phantom{0}6.79478505720    && 1.5183        \\
6  &&   \phantom{0}8.3129991933     && 1.5182        \\
7  &&   \phantom{0}9.829066969      && 1.5161        \\
8  &&    11.34534068       && 1.5163        \\
9  &&    12.8606086        && 1.5153        \\
10 &&    14.376072         && 1.5155        \\
11 &&    15.89096          && 1.5149        \\
12 &&    17.4060           && 1.5150        \\
13 &&    18.920            && 1.5147        \\
\end{tabular}
\end{center}
This  gives a refined asymptotic behaviour (see \cite[\S 5]{BEWZ})
\begin{equation}
\label{eq:asymptotic}
\theta_{k}^{-1} - \theta_{k-1}^{-1}  \simeq 1.515, \ \ k \leq 13 .
\end{equation}
This calls to question a recently suggested WKB approach \cite{RGMN}: the explanation proposed there gives the asymptotic spacing of magic angles  as $\theta_{ k}^{-1} - \theta_{k-1 }^{-1}  \simeq 1.47.$
In addition, the approach proposed in \cite{RGMN} would seemingly apply to all tunnelling potentials satisfying \eqref{eq:symmU} in the chiral model. However, the equidistant asymptotic distribution of magic angles is not stable under such perturbations, as we discuss in our next paragraph.

\smallsection{Lattice relaxations} It has been proposed in \cite{WG} that to take lattice relaxation effects into account, it is necessary to consider more general potentials in the continuum model. 
We therefore consider the simplest generalization of the tunnelling potential $U$  \eqref{eq:defU} that still satisfies \eqref{eq:symmU} \begin{equation}
\label{eq:tunnell}
  U_\mu ( z ) := U ( z ) 
+  \mu  \sum_{k=0}^2 \omega^k e^{\bar z \omega^k-z\bar \omega^k} , \ \ 
\mu \in \RR . 
\end{equation}

In this case, flat bands are still given as the eigenvalues of the respective Birman-Schwinger operator \eqref{eq:BSop}, with $U$ replaced by $U_{\mu}.$ However, the equidistant spacing of magic angles is no longer visible. See Fig.~\ref{fig:equidistant}, which also indicates that to understand the dynamics of $ \theta $ as 
$ \mu $ varies, complex values must be considered. If we 
abandon the {\em reality} requirement that $ \overline{ U ( \bar z ) } 
= U (z ) $ in \eqref{eq:symmU}, a generic $ U $ will not have any {\em real} magic angles.

\begin{figure}
\begin{center}
\includegraphics[width=7.5cm]{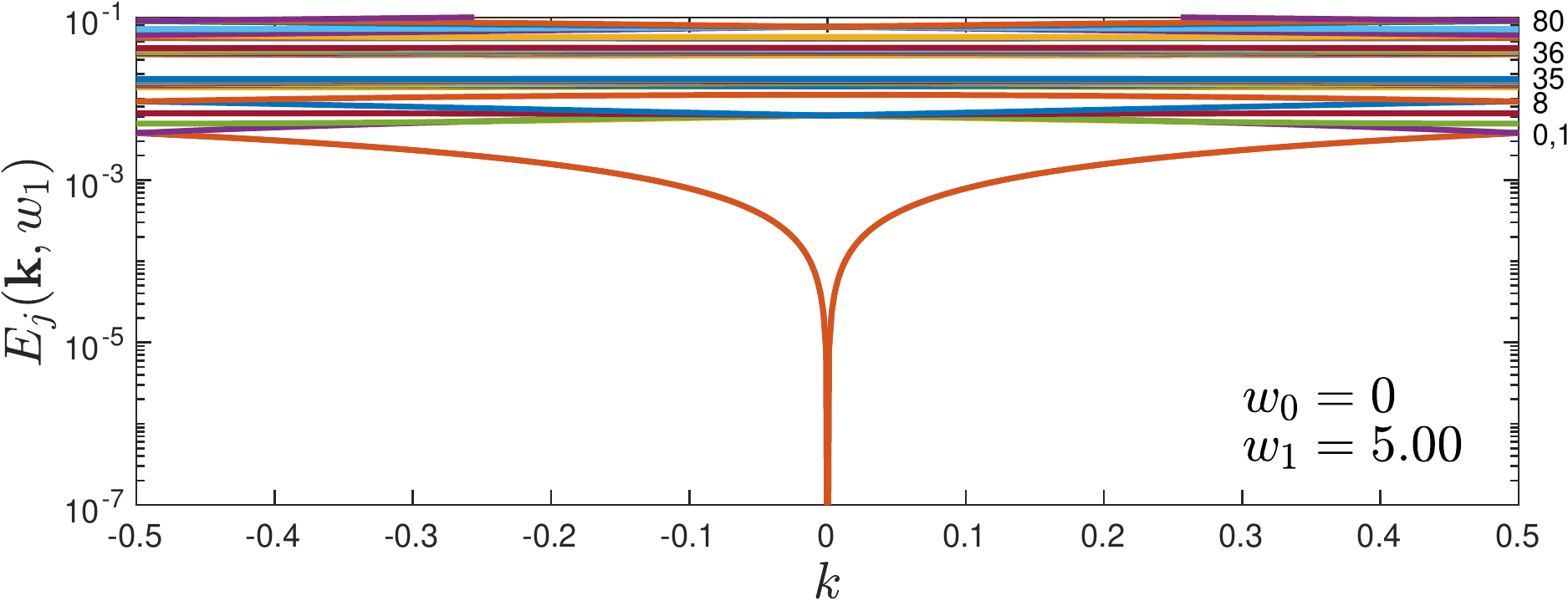}\\

\includegraphics[width=7.5cm]{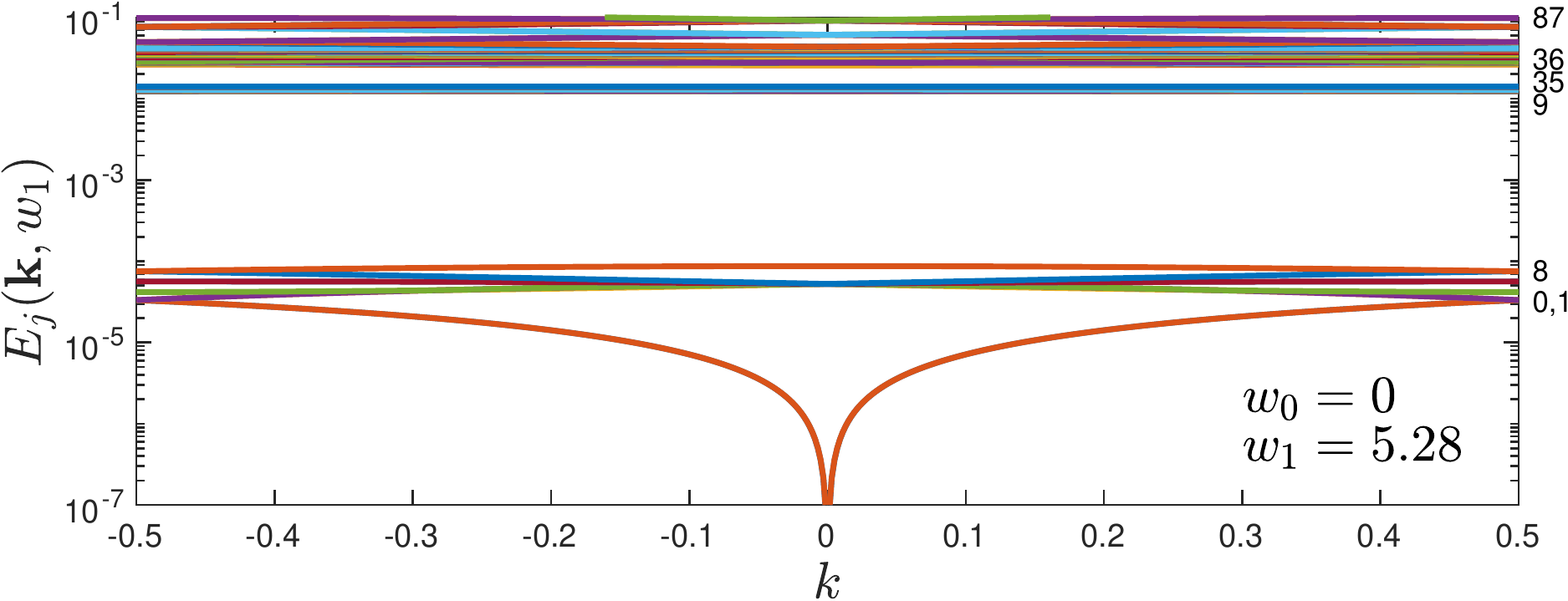}\\

\vspace*{5pt}
\hspace*{-5pt}\includegraphics[width=7.5cm]{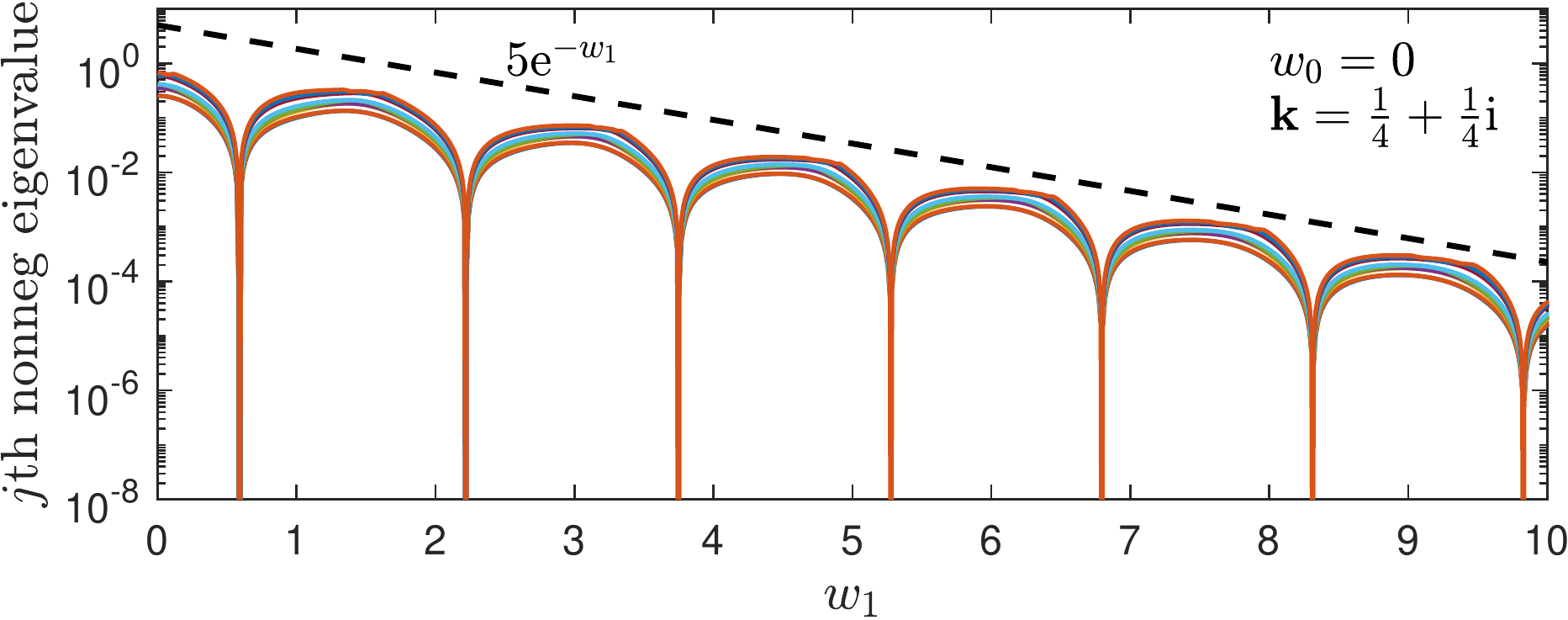}
\end{center}

\vspace*{-15pt}
\caption{\label{f:bands} Top: the smallest non-negative eigenvalues of $ H_{\mathbf k } ( 
(0,w_1),0 )$, $ w_1 = 5$, $ \mathbf k = k \omega/\sqrt 3$, $ -\frac12  \leq k \leq 
\frac12 $. (Selected $j$ values are given in the right margin.)
Middle:  The same, but with $w_1=5.28$ nearer the fourth magic $w_1$ value, showing the marked decrease of the nine lowest bands. Bottom: $E_j({\mathbf k},w_1)$ (log scale) for ${\mathbf k} = {1\over 4} + {1\over4}i$ and $0\le j\le 8$. We see \eqref{eq:sque} with $ c_1 = 1 $.}
\end{figure}
\smallsection{Point-localized states and exponential squeezing of bands} 
From \eqref{eq:asymptotic} we see that magic angles in the chiral model accumulate close to the zero twisting angle. However, aside from an accumulation of magic angles as $\theta \downarrow 0$, we also discover an exponential squeezing of eigenvalues of the chiral Floquet Hamiltonian $ H_{\mathbf k} ( (0,w_1), \varphi)$, 
$ \varphi \leq C/w_1 $, to zero. In the chiral model the low-lying bands for small angles become asymptotically flat:
if $ \{  E_j ( \mathbf k , w_1 ) \}_{ j =0}^{\infty} = 
\Spec ( H_{\mathbf k} ( (0,w_1),\varphi) ) \cap [ 0 , \infty ) $,
$ E_{j+1} \geq E_j $, then there are constants $c_0,c_1,c_2>0$ such that 
\begin{equation}
\label{eq:sque}  | E_j ( \mathbf k, w_1) | \leq c_0 e^{ - c_1w_1 } , \ \ \ \ 
 j  \leq c_2 w_1 , 
\end{equation}  see Fig.~\ref{f:bands}. In fact, numerical results here suggest that
$ c_1 = 1 $ and $ c_2 $ can be taken arbitrarily large.

\begin{figure}
\begin{center}
\includegraphics[width=4cm]{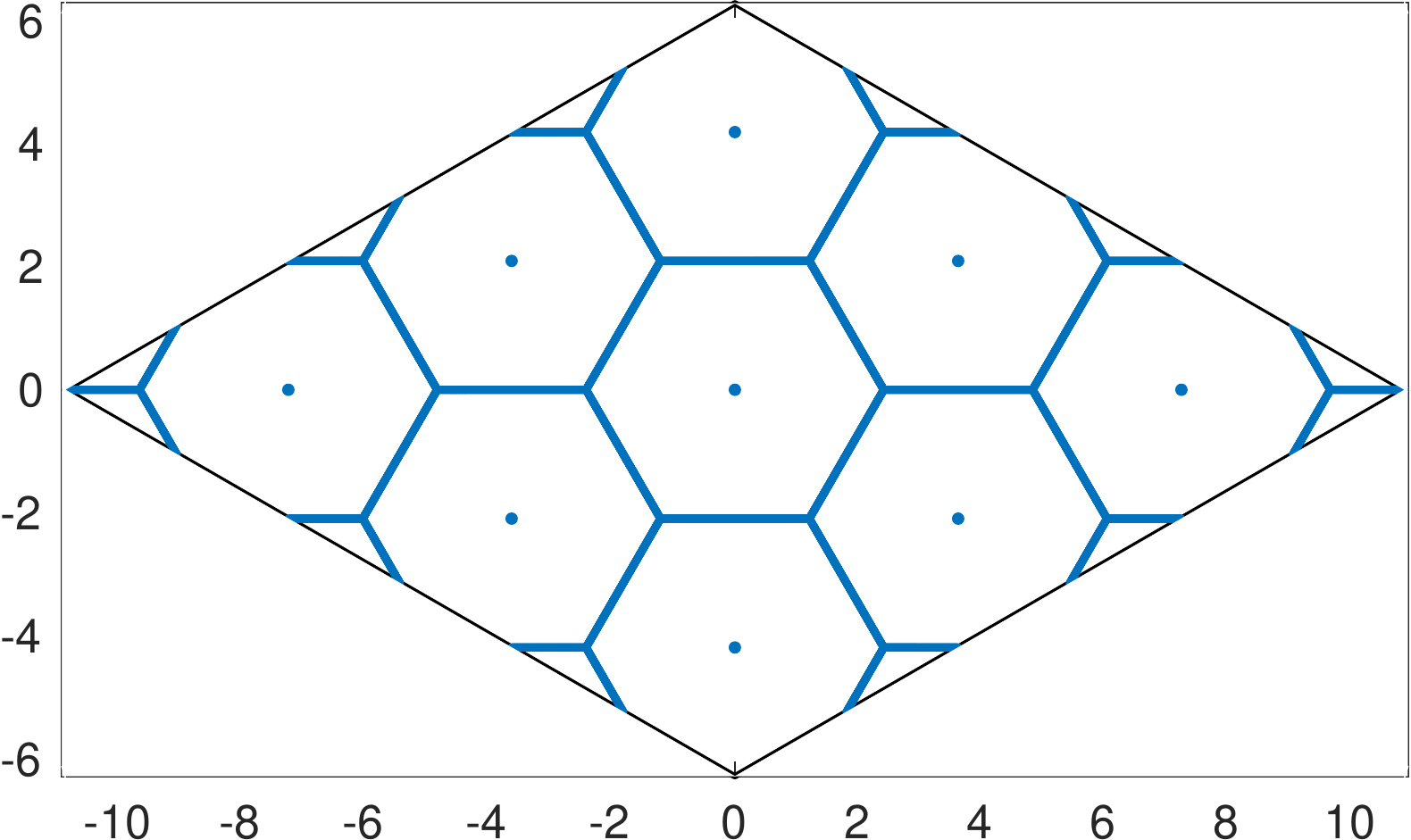} \quad 
\includegraphics[width=4cm]{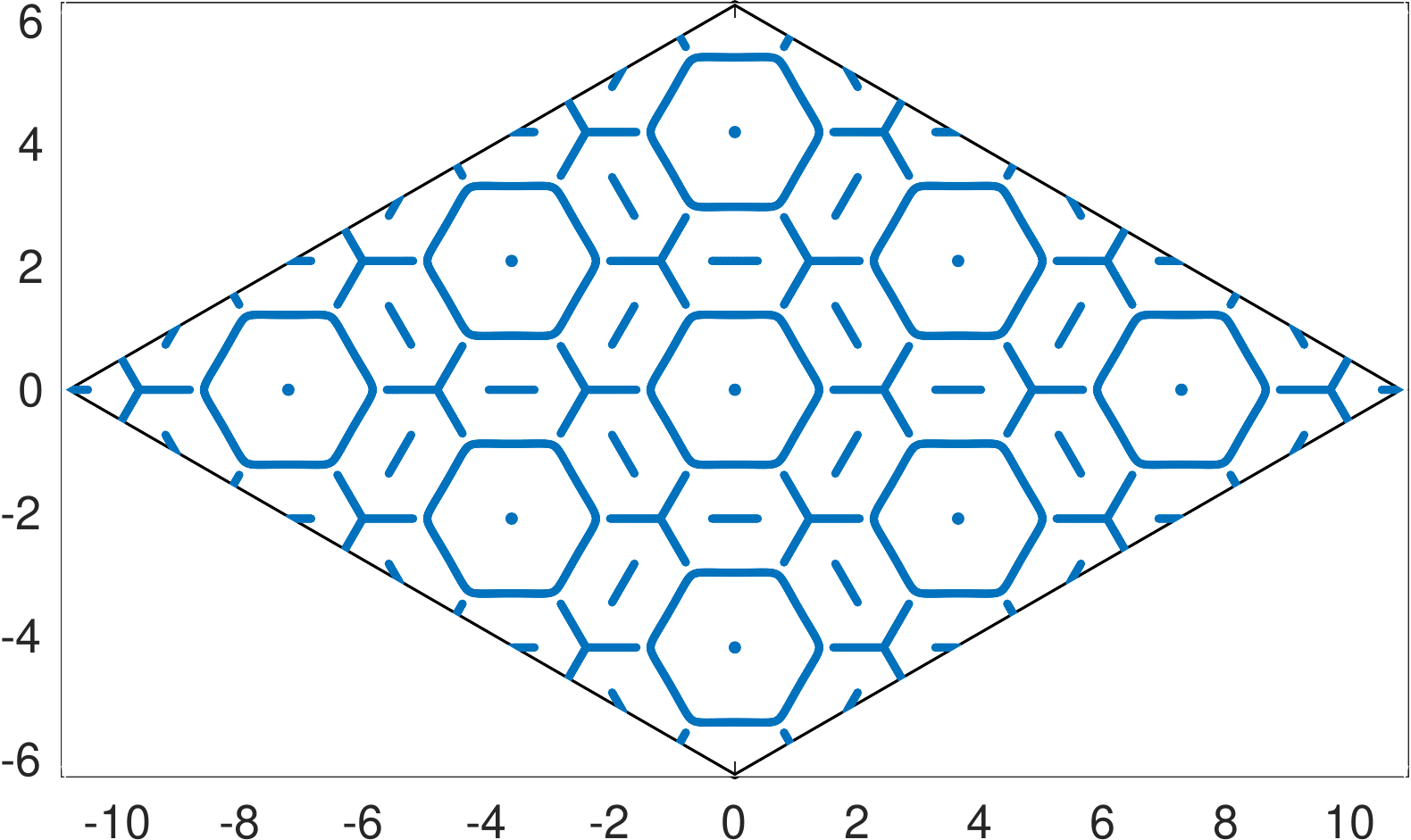}
\end{center}
\caption{\label{f:Horm} {Plots showing where the bracket $ i\{ q , \bar q \}$ is zero, 
where $q$ is the determinant of the symbol \eqref{eq:principalsymbol} 
with $U$ given by \eqref{eq:defU} (left) and by \eqref{eq:tunnell} with $\mu=-2$ (right).
The bracket is nonzero, except on a one-dim\-en\-sional graph and on a set of points.}}
\end{figure}

To understand the squeezing of low-lying eigenvalues theoretically, we use a semiclassical reformulation of the Hamiltonian of the chiral model. To study the nullspace of $ H_{\mathbf k} ( (0,w_1),0),$ it suffices to study effective 2-by-2 operators $D(w_1,0)$ in \eqref{eq:D}.  The principal Weyl symbol of $D(w_1,0)$ is given in terms of phase space variables $ z = x_1 + i x_2$ and $ 2  \zeta = \xi_1 - i \xi_2 $ by   
\begin{equation}\label{eq:principalsymbol}
\sigma(D(w_1,0))(z, \bar z, \bar \zeta)  = \begin{pmatrix} 2\bar \zeta & w_1 U(z) \\ w_1 U(-z) & 2\bar \zeta \end{pmatrix}. 
\end{equation}
\color{black}
Let $q$ be the determinant of this symbol. Then at points $(z^0,\zeta^0)$ in phase space at which $q (z^0,\zeta^0) = 0 $ and $\zeta^0 \neq 0$, we then find that the Poisson-bracket
$\{ q , \bar q \} (z^0,\zeta^0) \neq 0$ is non-vanishing.

An adaptation of H\"ormander's bracket condition to the analytic case \cite[Theorem 1.2]{dsz} implies that for $ z^0 $ away from the set
shown in 
Fig.~\ref{f:Horm}, 
 there exist smooth functions $v_h$ such that $|  D(w_1,h)  v_h (z ) | \leq  e^{ - \frac c  h } $ ($ h \leq C/w_1 $), which are localized to $z^0$, i.e., $v_h (z^0) = 1$ and $ |  v_h ( z ) | \leq  e^{ - \frac{c}{h} | z- z^0|^2 },$ see \cite[Theorem 3]{BEWZ}.

However, the above argument does not apply to the Hamiltonian of the full continuum model \eqref{eq:Hamiltonian} while in the anti-chiral model
($ w_1 = 0 $) there are no localized modes at $z^0 $ with $ U ( z^0 ) U ( -z^0 ) \neq 0 $ -- see Theorem \ref{theo:Jens} in the supplementary material.
This is confirmed by numerics:  
 Fig.~\ref{fig:nosqueeze} shows that there is no exponential squeezing of bands.

\begin{figure}

\hspace*{-2em}\includegraphics[width=8.5cm]{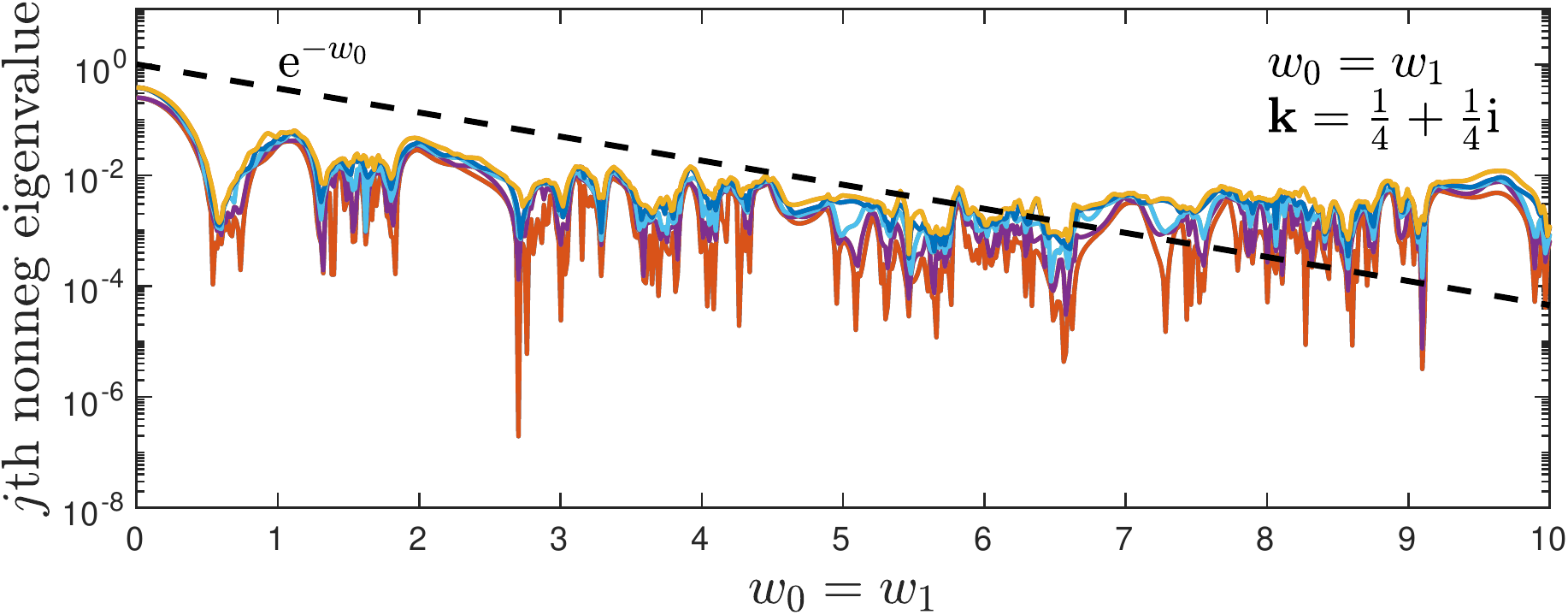}

\vspace*{.5em}
\hspace*{-2em}\includegraphics[width=8.5cm]{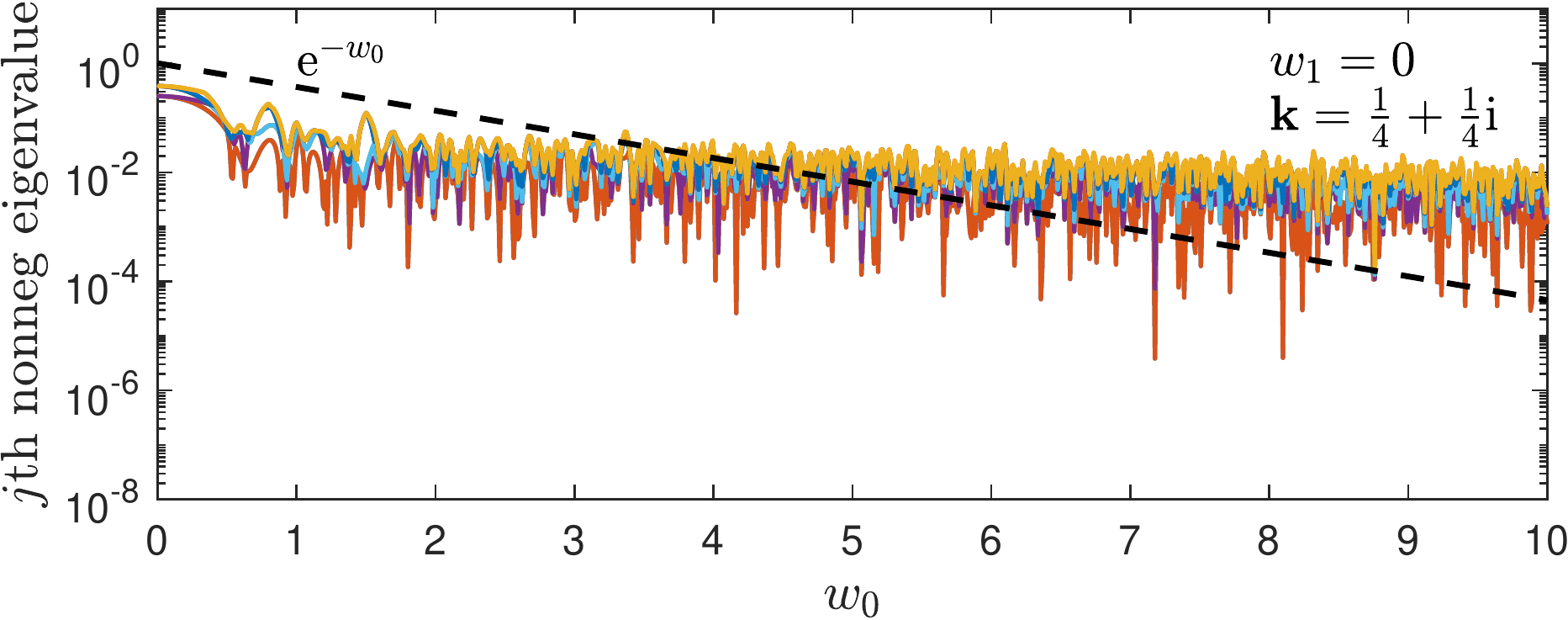}

\caption{Top:  For $ w_0 = w_1 \gg 1 $, $ \varphi =0 $,
no ex\-ponen\-tial rate of squeezing ($1 \le j \le 10$). The bands show a squeezing effect near $ w_0=w_1 \simeq 0.586 $. Bottom: no squeezing
for $ w_1 = 0$ and $ w_0\gg 1 $ consistent with the absence of localized modes and perfectly flat bands.}
\label{fig:nosqueeze}
\end{figure}

\smallsection{Symmetries and protected states at zero} Magic angles in the chiral continuum model are the reciprocals of coupling constants $w_1$ for which $0$ is in the spectrum of all $H_{\mathbf k}((0,w_1),0)$. 

We show that the symmetries of the continuum model with particle-hole symmetry imply that the zero energy level is protected in the spectrum of the Hamiltonian for all $w=(w_0,w_1).$

We start by stating the most straightforward symmetries:  translational symmetry
\begin{equation}
 {\mathscr L}_{\mathbf a } :=  \mathscr U  L_{\mathbf a } , \ \ \
 \mathscr U  := \operatorname{diag}( \omega,  1 , \omega , 1), \ \ 
\mathbf a = \tfrac{ 4}3 \pi i \omega^\ell , \end{equation}
and rotational $C_3$-symmetry 
\[ \mathscr C \mathbf u ( z) = {\rm{diag}} ( I_{\CC^2}  , \bar \omega I_{\CC^2} ) 
\mathbf u ( \omega z ).
 \]
From $ \mathscr C \mathscr L_{\mathbf a } = \mathscr L_{\bar \omega \mathbf a } \mathscr C$, we can construct an action, for $\Gamma_3 = (\Gamma/3)/\Gamma$, of the Heisenberg group over $\ZZ_3$:
\begin{equation}
\begin{gathered}  G := \Gamma_3 \rtimes \ZZ_3 ,  \ \ \ ( \mathbf a , k ) \cdot ( \mathbf a' , \ell ) = 
( \mathbf a + \bar \omega \mathbf a' , k + \ell ) ,
\\  ( \mathbf a, \ell ) \cdot \mathbf u := 
\mathscr L_{\mathbf a } \mathscr C^\ell \mathbf u . 
\end{gathered}
\end{equation}
Since there are eleven irreducible representations, see \cite[\S 2.2]{BEWZ}, we can decompose $L^2(\CC/\Gamma; \CC^4)$ into eleven orthogonal subspaces. For $(k,p) \in \ZZ_3^2$ we have that $9$ of these spaces, $L^2_{k,p}$, are characterized by the action $\mathscr L_{\mathbf a }|_{L^2_{ k , p }} = 
\omega^{ k ( a_1 + a_2 ) } $ and $\mathscr C|_{L^2_{ k , p }} = \bar \omega^p.$
It is then a simple observation that for $w=(w_0,w_1)=0$ we have that $H(w,0)e_i=0$ for $i=1,\ldots,4$ and
$e_1 \in L^2_{1,0}$, $e_2 \in L^2_{0,0}$, $e_3 \in L^2_{1,1}$, $e_4 \in L^2_{0,1}$
are all in different subspaces. To see that these elements are protected by symmetries, we require more subtle symmetries, see \cite{K19,L19} which we discuss in detail in the supplementary material: Mirror symmetry $\mathscr M : L^2_{k,p} \to L^2_{-k+1,-p+1}$, Lemma \ref{l:mir} 
\[ \mathscr M \mathbf u ( z ) := \begin{pmatrix} 0 & \sigma_1 \\ \sigma_1 & 0 \end{pmatrix} \mathbf u ( \bar z ),  
\] 
and $\mathcal P \mathcal T$-symmetry $\mathcal P \mathcal T : L^2_{k,p} \to L^2_{k,-p+1}$, Lemma \ref{lemm:PT}
\[ \mathcal P \mathcal T \mathbf u ( z ) = \begin{pmatrix} 0 & I_{\CC^2} \\
\ I_{\CC^2} & 0 \end{pmatrix} \overline{\mathbf u ( - z )}.
\]
All of the above symmetries commute with $  H ( w, \varphi)  $, also for $\varphi \neq 0$. 
Only when setting $\varphi=0$, the Hamiltonian exhibits in addition particle-hole symmetry $\mathscr S: L^2 _{k,p} \to  L^2_{-k+1,p} $, such that $\mathscr S H(w,0) = - H(w,0) \mathscr S$, where
\[  \mathscr S \mathbf u ( z ) = \begin{pmatrix} \sigma_2 & 0 \\
0  & \sigma_2 \end{pmatrix} \mathbf u(-z).\]

From the application of the last three symmetries, we find
$ L^2 _{k,p} \xrightarrow{ \mathscr S } L^2_{-k+1,p} \xrightarrow{ \mathscr M } L^2_{k,-p+1}  \xrightarrow{ \mathcal P \mathcal T } L^2_{k,p}$, see Lemma \ref{l:chir},
which shows that $\Spec_{L^2_{k,p} } ( H  ) = - \Spec_{L^2_{k,p}. 
}( H  ).$
Recalling the kernel for $ w = ( 0, 0 ) $, we conclude that $\ker_{ {L^2_{k,p} }} H( w, 0 ) \neq \{ 0 \} $,
$ k,p \in \{ 0  , 1 \}$, $ w \in \RR^2$.

\smallsection{Absence of flat bands} Unlike for the chiral model, which exhibits infinitely many flat bands, perfectly flat bands at zero are absent once AA-coupling is switched on. To understand this in the extremal anti-chiral case, we follow an idea of Thomas~\cite{thomas} for Schr\"odinger operators. 
First, observe that $0 \in \Spec(H_{\mathbf{k}})$ is 
equivalent to $\ker  Q_{\mathbf k}  \oplus \ker  Q_{\mathbf k}^*  \neq \{ 0 \}$
(see Theorem~\ref{theo:noflatbands} for details), where
\[ Q_{\mathbf k}(w_0):=
\begin{pmatrix}
  w_0e^{i  \varphi/2} V(z,\bar z)  & (2D_{z} +\mathbf{ \bar k})    \\
 (2D_{\bar z} +\mathbf{ k}) & w_0 e^{i  \varphi/2}\overline{ V(z,\bar z)}
\end{pmatrix}.\]
By squaring the operators, we find the identity
 \begin{equation}
\begin{split}
\label{eq:identity}
 Q_{\mathbf k}( Q_{\mathbf k}  + V_{11}) &= ( D_x + (k_1 ,k_2 ) )^2 I_{\CC^2}  + V_{12}   \\
 Q_{\mathbf k}^* ( Q_{\mathbf k}^* + V_{21} )  &=  ( D_x + (k_1 ,k_2 ) )^2 I_{\CC^2}  + {V}_{22},  
 \end{split}
\end{equation}
for the momentum operator $D_x = \tfrac 1 i ( \partial_{x_1} , \partial_{x_2} ) $ and auxiliary potentials $V_{ij} = V_{ij}  (\beta ).$ If we then \emph{complexify} the quasi-momentum $k_1$, self-adjointness of the momentum operator implies 
$(( D_x + (k_1 ,k_2 ))^2 )^{-1} = \mathcal O_{L^2 \to L^2 } ( |\Im k_1|^{-2} )$ such that by \eqref{eq:identity} $\ker Q_{\mathbf k} 
\oplus \ker Q_{\mathbf k}^* = \{ 0 \} $ for $\ | \Im k_1 | \gg 1.$ Thus, if there was a flat band such that $0 \in \Spec(H_{\mathbf k}) , \ \ (k_1, k_2) \in 
\RR^2$ this would imply that $0 \in \Spec(H_{\mathbf k}),$ $\mathbf k= (k_1, k_2) \in \CC \times \RR$ and thus
$\ker Q_{\mathbf k} 
\oplus \ker Q_{\mathbf k}^* \neq \{ 0 \}$ $(k_1, k_2) \in \CC \times \RR$ which is impossible. Details are provided in Theorem \ref{theo:noflatbands} of the supplementary materials.

\smallsection{Acknowledgements} We would like to thank Mike Zaletel for bringing \cite{TKV} to our 
attention and Alexis Drouot, Grisha Tarnopolsky and 
Ashvin Vishwanath for helpful discussions.
SB~gratefully acknowledges support by
the UK Engineering and Physical Sciences Research Council (EPSRC)
grant EP/L016516/1 for the University of Cambridge Centre for Doctoral
Training, the Cambridge Centre for Analysis. ME and MZ were partially 
supported by the National Science Foundation under the grants DMS-1720257 and DMS-1901462, respectively. 
JW was partially supported by the Swedish Research Council grants 2015-03780 and 2019-04878.

\clearpage

\onecolumngrid

\section{Supplementary Material}

We provide additional mathematical details on four parts that are only sketched in the letter.
\begin{itemize}
\item A derivation of the continuum model \eqref{eq:Hamiltonian} in the article.
\item A careful and extensive analysis of symmetries and protected states of the Hamiltonian \eqref{eq:Hamiltonian} of the continuum model.
\item The non-existence of flat bands in the anti-chiral model.
\item No exponential squeezing in the anti-chiral model.
\end{itemize}

\section{Continuum model}
Let 
\begin{equation}
\begin{split}
\sigma_0 &= \begin{pmatrix} 1 & 0 \\ 0 & 1 \end{pmatrix}, \quad \sigma_1 = \begin{pmatrix} 0 & 1 \\ 1 & 0 \end{pmatrix}, \quad
\sigma_2 = \begin{pmatrix} 0 & -i \\ i & 0 \end{pmatrix}, \quad \sigma_3 = \begin{pmatrix} 1 & 0 \\0  & -1 \end{pmatrix}
\end{split}
\end{equation}
be the Pauli matrices. The continuum model for a single valley of twisted bilayer graphene is described by an effective Hamiltonian on $\RR^2$,
\begin{equation}
\begin{split}
\mathcal H'(\alpha,\varphi)=\begin{pmatrix} - i \sigma_{\varphi/2}\nabla & T(k_{\theta}x) \\ T(k_{\theta}x)^* & - i \sigma_{-\varphi/2}\nabla  \end{pmatrix},
\end{split}
\end{equation}
with $\nabla = (\partial_{x_1}, \partial_{x_2})$ and $\sigma_{\varphi/2} = e^{-i \frac{\varphi}{4} \sigma_3} (\sigma_1, \sigma_2) e^{i \frac{\varphi}{4} \sigma_3}$ and tunneling potentials
\[T(x, \alpha) = \sum_{j=1}^3 T_j(\alpha) e^{-i  q_j \dot (x_1,x_2)^t } \]
with $q_1 = (0,-1), q_{2,3}= \tfrac{1}{2}(\pm \sqrt{3},1)$ such that for hopping parameters $\alpha_0, \alpha_1$
\[ T_{j+1}(\alpha)= \alpha_0 \sigma_0 + \alpha_1 (\cos(\phi j) \sigma_1+ \sin(\phi j) \sigma_2),\]
where $\phi=2\pi/3$ and $k_{\theta} = 2k_D \sin(\theta/2)$ is the Moir\'e wave vector.
We then rescale coordinates by $k_{\theta}$ to obtain a new Hamiltonian with $w:=\alpha/k_{\theta}$: 
\begin{equation}
\label{eq:stand_form}
\mathcal H(w, \varphi) =\begin{pmatrix} - i \sigma_{\varphi/2}\nabla & T(x; w) \\ T(x;w)^* & - i \sigma_{-\varphi/2}\nabla  \end{pmatrix}.
\end{equation}
Switching to complex coordinates $z=x_1+ix_2$ and conjugating by $U=\operatorname{diag}(1,\sigma_1,1)$ then yields the Hamiltonian 
\begin{equation}\label{eq:conjugationrelation}
 H(w, \varphi)=U \mathcal H(w,\varphi) U
 \end{equation}
stated in \eqref{eq:Hamiltonian}, studied in this article.
If instead of conjugating by $U$, we conjugate, for $\lambda=e^{i \pi/4}$, by 
\[\mathcal V=\begin{pmatrix} i \lambda & 0 & 0 & 0 \\ 0 & 0 & 0 & -\bar \lambda \\ 0 & 0 & i \lambda & 0 \\ 0 & - \bar \lambda & 0 &0 \end{pmatrix},\]
we obtain an operator that has a particularly simple form for the anti-chiral model, $w_1=0$:
\begin{equation}
\begin{split}
\label{eq:Q}
\mathscr H(w,\varphi)&:=k_{\theta}^{-1} \mathcal V H(w,\varphi)\mathcal V^*\\
&=\begin{pmatrix} 0&  Q(w_0,\varphi)^* \\ Q(w_0,\varphi) & 0\end{pmatrix}
\end{split}
\end{equation}
with 
\begin{equation}
\label{eq:Q2}
Q(w_0,\varphi):=
\begin{pmatrix}
  w_0 V(z,\bar z)  & 2e^{i  \varphi/2}D_{z}    \\
 2e^{i  \varphi/2}D_{\bar z} & w_0 \overline{ V(z,\bar z)}
\end{pmatrix}.
\end{equation}

\section{Symmetries}

In this section we provide additional details on the symmetries that the Hamiltonian satisfies.
We recall that the potentials in the Hamiltonian $H(w,\varphi)$ of the continuum model \eqref{eq:Hamiltonian} satisfy the symmetry relations
\begin{equation}
\label{eq:symU}
U ( z + \mathbf a ) = \bar \omega U ( z ) , \ \ \mathbf a = \tfrac 43 \pi i \omega^\ell, \ \ell = 1,2 , \ \ U ( \omega z ) = \omega U ( z ) , \ \ 
\overline{ U ( \bar z ) } = U ( z ) , 
\end{equation}
and hence, with the same $ \mathbf a $, for $V(z)=2\partial_z U(z)$
\begin{equation}
\label{eq:symV}
V ( z + \mathbf a ) = \bar \omega V ( z ) , \ \ 
V ( \omega z ) = V ( z ) , \ \  \overline{ V ( z ) } = V ( - z ) , \ \ 
{ V ( \bar z ) } = V (- z ) . 
\end{equation}
We can generalize these potentials by considering, for $n \in \mathbb Z$, the functions
\begin{equation}
\begin{split}
f_n(z)&:=\sum_{k=0}^2 \omega^k e^{\frac{n}{2}(z \bar \omega^k- \bar z \omega^k)} \text{ and }\\
g_n(z)&:= \sum_{k=0}^2 e^{\frac{n}{2}(z \bar \omega^k- \bar z \omega^k)}.
\end{split}
\end{equation}
 Then, $f_n$ and $g_n$ satisfy the rotational symmetries $f_n(\omega z)=\omega f_n(z)$ and $g_n(\omega z)=g_n(z)$, and in addition, the translational symmetries
\begin{equation}
\begin{split}
 f_n(z+\mathbf a) &= \bar \omega^n f_n(z), \quad  \mathbf a = {\textstyle \frac{4}{3}} \pi i \omega^{\ell},\quad  \ell \in \ZZ_3, \\
 g_n(z+\mathbf a) &= \bar \omega^n g_n(z), \quad  \mathbf a = {\textstyle \frac{4}{3}} \pi i \omega^{\ell},\quad  \ell \in \ZZ_3.
\end{split}
\end{equation}
Hence, we may replace the potentials $U,V \in C^{\infty}(\CC/\Gamma)$ in the continuum model for suitable $\alpha_n, \beta_n \in \RR$ by
\begin{equation}
\label{eq:gen_pot}
U(z) = \sum_{n = 3\ZZ+1} \alpha_n g_n(z) \text{ and } V(z) = \sum_{n = 3\ZZ+1} \beta_n f_n(z).
\end{equation}

The symmetries of the Hamiltonian $H(w,\varphi)$ were thoroughly studied in \cite{K19,L19} for the continuum model with shifted Dirac points. In particular, the mirror symmetry that we study in Lemma \ref{l:mir} 
was pointed out in \cite{K19}. When neglecting the shift of Dirac points
(putting $ \varphi=0$), we find an additional particle-hole symmetry in our model; see Lemma \ref{l:chir}. 

\medskip

The Hamiltonian $ H ( w,\varphi ) $ is periodic with respect to 
$ \Gamma  := 4  \pi i ( \omega \ZZ \oplus \omega^2 \ZZ ) $. More 
interesting symmetries are provided by an action of the Heisenberg group
over $ \ZZ_3 $:
\begin{equation}
\label{eq:defG}   
\begin{gathered} G_3 \simeq \ZZ^2_3 \rtimes \ZZ_3, \ \ 
\ZZ_3 \ni k  : \mathbf a \to \bar \omega^k \mathbf a , \ \ 
( \mathbf a , k ) \cdot ( \mathbf a' , \ell ) = ( \mathbf a + \bar \omega^k 
\mathbf a' , k + \ell ), \\
\mathbf a \in \Gamma_3 := \Gamma/(\Gamma/3) \simeq \ZZ_3^2 , \ \ F : 
\mathbf a = \tfrac 43 \pi ( a_1 i \omega + a_2 i \omega^2 ) \mapsto
( a_1 , a_2 ) \in \ZZ^2_3 , \\
 \bar \omega \mathbf a =  F^{-1} M F \mathbf a , \ 
\ \ M =  \begin{pmatrix} -1 & 1 \\ -1 & 0 \end{pmatrix} . 
\end{gathered}
\end{equation}

From \eqref{eq:symU} and \eqref{eq:symV} we see that for $D$ and $C$ defined in~\eqref{eq:D},
and $ L_{\mathbf a } \mathbf v ( z ) :=\mathbf v ( z + \mathbf a ) $,  we have
\[  D(w_1,\varphi) L_{ \mathbf a } = 
L_{\mathbf a } \begin{pmatrix} 2 e^{i\varphi/2}   D_{\bar z } &  \omega w_1 U 
\\ \bar \omega w_1 U ( - \bullet ) & 2   e^{-i\varphi/2}D_{\bar z }  \end{pmatrix} = 
 \begin{pmatrix}  \omega & 0 \\ 0 & 1 \end{pmatrix} 
L_{\mathbf a } D(w_1,\varphi) \begin{pmatrix} \bar  \omega & 0 \\ 0 & 1 \end{pmatrix} , \]
and 
\[ C ( w_0 ) L_{\mathbf a } = w_0 L_{\mathbf a } 
\begin{pmatrix} 0 & \omega V \\ \bar \omega V ( - \bullet ) & 0 \end{pmatrix}
= \begin{pmatrix}  \omega & 0 \\ 0 & 1 \end{pmatrix} 
L_{\mathbf a } C(w_0) \begin{pmatrix} \bar  \omega & 0 \\ 0 & 1 \end{pmatrix} .\]

Hence if we define
\begin{equation}
\label{eq:defLa}  {\mathscr L}_{\mathbf a } :=  \mathscr U  L_{\mathbf a } , \ \ \
 \mathscr U  := \begin{pmatrix} \omega & 0 & 0 & 0 \\
0 & 1 & 0 & 0 \\
0 & 0 & \omega & 0  \\
0 & 0 & 0 & 1 \end{pmatrix}, \ \ 
\mathbf a = \tfrac{ 4}3 \pi i \omega^\ell , \end{equation}
then 
\[  \mathscr L_{\mathbf a } H(w,\varphi) = H(w,\varphi) \mathscr L_{\mathbf a } . \]
We extend this to an action of $ \Gamma_3 $ by putting
\[ \mathscr L_{\mathbf a } \mathbf u  = \mathscr U^{a_1 + a_2 } \mathbf 
u ( z + \mathbf a ) , \ \ \ \mathbf a := \textstyle{\frac{4}{3}}
\pi i ( a_1 \omega + a_2 \omega^2 ). \]
Using $ U ( \omega z ) = \omega U ( z ) $ and $ V ( \omega z ) = V ( z ) $,
we also see that the unitary transformation
\[  \mathscr C : L^2 ( \CC ; \CC^4 ) \to L^2 ( \CC ; \CC^4 ) , \ \ \
 \mathscr C \mathbf u ( z ) := \begin{pmatrix} 1 & 0 & 0 & 0 \\
 0 & 1 & 0 & 0 \\
 0 & 0 &  \bar \omega & 0 \\
0 & 0 & 0 &  \bar \omega \end{pmatrix} \mathbf u (  \omega z )
\]
commutes with $ H ( w,\varphi ) $.
Since $   \mathscr C \mathscr L_{\bar \omega \mathbf a }
= \mathscr L_{ \mathbf  a }  \mathscr C $, we see that we obtain 
a unitary action of  $G_3 $ commuting with $ H ( w,\varphi ) $.

We now list additional symmetries.
The first one is a simple symmetry in the parameters $w=(w_0,w_1)$.
\begin{lemm}[$w$-symmetry]
If $\mathscr Q : =\operatorname{diag}(i,-i,i,-i)$, then 
\[  \mathscr Q H( w,\varphi ) \mathscr Q^* =H ( -w ,\varphi), \ \ \ \mathscr L_{\mathbf a } \mathscr Q = \mathscr Q 
\mathscr L_{\mathbf a }, \ \ \mathscr C \mathscr Q = \mathscr Q \mathscr C, 
 \]
for $ \mathbf a \in \Gamma_3 $.
\end{lemm}

\begin{lemm}[Mirror symmetry]
\label{l:mir}
For the unitary involution
\[ \mathscr M \mathbf u ( z )  := \begin{pmatrix} 0 & \sigma_1 \\ \sigma_1 & 0 \end{pmatrix} \mathbf u ( \bar z ) 
,  \ \ \ \sigma_1 := \begin{pmatrix} 0 & 1 \\ 1 & 0 \end{pmatrix} , 
\]
we have 
\begin{equation}
\label{eq:H2M} 
\mathscr M H ( w ,\varphi) = H ( w,\varphi ) \mathscr M , \ \ \ \ 
\mathscr M \mathscr C = \bar \omega \mathscr C^{-1} \mathscr M , \ \ \ \
\mathscr M  \mathscr L_{\mathbf a }= \omega^{a_1 + a_2}  \mathscr L_{\bar {\mathbf a } } \mathscr M. 
\end{equation}
%where
%\begin{equation}
%\label{eq:defJ}
%\begin{gathered} 
% J \mathbf a := - \bar {\mathbf a } , \ \ \  J : \Gamma_3 \to 
%\Gamma_3 , \ \ \  F J F^{-1} ( a_1 , a_2  ) = (a_2 , a_1 ) ,  \\
% F :  \Gamma_3 \to \ZZ_3^2 , \ \ \ F  ( \tfrac 43 \pi i ( a_1 \omega + a_2\omega^2 )) := ( a_1, a_2 ) . \end{gathered} \end{equation}
\end{lemm}
\begin{proof}
We have
\[ \begin{pmatrix} 0 & \sigma_1 \\ \sigma_1 & 0 \end{pmatrix} H ( w,\varphi )  \begin{pmatrix} 0 & \sigma_1 \\ \sigma_1 & 0 \end{pmatrix}  = 
\begin{pmatrix} \sigma_1 C ( w_0) \sigma_1 &  \sigma_1 D ( w_1,\varphi ) 
\sigma_1 \\ (\sigma_1 D ( w_1,\varphi ) 
\sigma_1 )^* & \sigma_1 C ( w_0 ) \sigma_1 \end{pmatrix}
\]
and
\[ \begin{split} \sigma_1 C ( w_0 ) \sigma_1 = w_0 \begin{pmatrix}
0 & V ( - z ) \\ V ( z ) & 0 \end{pmatrix} & = 
w_0 \begin{pmatrix} 0 & V ( \bar z ) \\
V ( - \bar z ) & 0 \end{pmatrix} 
, \\
( \sigma_1 D ( w_1 ) \sigma_1 )^* = \begin{pmatrix} 
2 e^{-i \varphi/2}D_z & w_1\overline {U ( z ) } \\
w_1 \overline {U ( - z )} & 2 e^{i \varphi/2}D_z \end{pmatrix} 
& =  \begin{pmatrix} 
2 e^{-i \varphi/2}D_z & w_1 {U ( \bar z ) } \\
w_1  {U ( - \bar z )} & 2 e^{i \varphi/2} D_z \end{pmatrix}, \end{split}
 \]
where we used  $ V ( - z ) = V ( \bar z ) $ and $ U ( \bar z ) =
\overline { U (  z ) } $. Hence, if we define $ \Sigma \mathbf u ( z ) :=
\mathbf u ( \bar z ) $ then,
\[  \sigma_1 C( w_0 ) \sigma_1 = \Sigma C ( w_0 ) \Sigma , \ \ 
( \sigma_1 D( w_1,\varphi ) \sigma_1 )^*= \Sigma D ( w_1, \varphi ) \Sigma . \]
This gives $ \mathscr M H ( w,\varphi ) \mathscr M = H ( w,\varphi ) $. 
Now,
\[ \begin{split} \mathscr M \mathscr C \mathbf u ( z ) & = 
\begin{pmatrix} 0 & \sigma_1 \\ \sigma_1 & 0 \end{pmatrix} 
\begin{pmatrix}   I_{\CC^2}  & 0 \\ 0 & \bar \omega I_{\CC^2}  \end{pmatrix}  \mathbf u ( \omega \bar z ) 
= \begin{pmatrix} 0 & \bar \omega \sigma_1 \\ \sigma_1 & 0 \end{pmatrix}
 \mathbf u\left( \overline { \bar \omega z } \right) \\
&= 
\begin{pmatrix} \bar \omega  I_{\CC^2} & 0 \\
0 &  I_{\CC^2}  \end{pmatrix} \begin{pmatrix} 0 & \sigma_1 \\ \sigma_1 & 0 
\end{pmatrix} \mathbf u \left( \overline { \bar \omega z } \right)   =
\bar \omega \mathscr C^{-1} \mathscr M \mathbf u ( z ) . \end{split} \]
Finally, in the notation of \eqref{eq:defLa} and with 
$ F ( \mathbf a ) = ( a_1, a_2 ) $ (see \eqref{eq:defG}),
\[ \begin{split} \mathscr M  \mathscr L_{\mathbf a } \mathbf u ( z ) & = 
\begin{pmatrix} 0 & \sigma_1 \\ \sigma_1 & 0 
\end{pmatrix} \mathscr U^{a_1 + a_2}  \mathbf u ( \bar z + \mathbf a ) =
\begin{pmatrix} 0 & 0 & 0 & 1 \\
0 & 0 & \omega^{a_1 + a_2}  & 0 \\
0 & 1 & 0 & 0 \\
\omega^{a_1+a_2}  & 0 & 0 & 0 \end{pmatrix} \mathbf u \left( \overline
{ z + \bar{ \mathbf a }} \right) \\
& = 
\omega \mathscr U^{-(a_1+a_2) } \begin{pmatrix} 0 & \sigma_1 \\ \sigma_1 & 0 
\end{pmatrix} \mathbf u \left( \overline { z + \bar{\mathbf a }} \right) 
= \omega \mathscr L_{\bar {\mathbf a } } \mathscr M \mathbf u ( z ) ,
\end{split} 
\]
since $ F ( \bar {\mathbf a} ) = ( - a_2, -a_1 ) $. 
\end{proof}
The next symmetry we shall discuss for the continuum is the $ \mathcal P \mathcal T $ symmetry.

\begin{lemm}[$\mathcal P \mathcal T$-symmetry]
\label{lemm:PT}
Define two commuting involutions, anti-linear and linear, respectively:
\begin{equation}
\label{eq:defPT} \mathcal T \mathbf u ( z ) := \overline {\mathbf u ( z )}, \ \ \ \mathcal P \mathbf u := \begin{pmatrix} 0 & I_{\CC^2} \\
\ I_{\CC^2} & 0 \end{pmatrix} \mathbf u ( - z ) . \end{equation}
Then
\begin{equation}
\label{eq:H2PT}   \mathcal P \mathcal T H(w,\varphi) =H(w,\varphi) \mathcal P
\mathcal T , \ \ \  \mathcal P \mathcal T \mathscr C = \omega \mathscr C \mathcal P \mathcal T,
\ \ \
\mathcal P \mathcal T \mathscr L_{\mathbf a } = \mathscr L_{-\mathbf a} 
\mathcal P  \mathcal T , \ \ \
 \mathbf a \in \Gamma_3 .
\end{equation}
\end{lemm}
\begin{proof}
The first identity in \eqref{eq:H2PT} follows from immediately from
 $\overline{V(-z)}= V(z)$. To see the second one, we write 
\[
\begin{split}
\mathcal P \mathcal T \mathscr C \mathbf u(z) = \begin{pmatrix} 0 & \ \omega I_{\CC^2}  \\ I_{\CC^2}  & 0 \end{pmatrix} \overline{\mathbf u(-\omega z)}
= \omega 
 \begin{pmatrix} 0 &  I_{\CC^2}  \\ \ \bar \omega I_{\CC^2}  & 0 \end{pmatrix}
\overline{\mathbf u(-\omega z)} = 
\omega \mathscr C \mathcal P \mathcal T \mathbf u ( z ) . \end{split} \]
Finally, for $ \mathbf a = \frac 43 \pi \omega^\ell$, $ \ell = 1,2 $, 
\[  \mathcal P \mathcal T \mathscr L_{\mathbf a } \mathbf u ( z ) = 
\begin{pmatrix} 0 & 0 & \bar \omega & 0 \\
0 & 0 & 0 & 1 \\
\bar \omega  & 0 & 0 & 0  \\
0 & 1 & 0 & 0 \end{pmatrix} \overline{\mathbf u ( - (z - \mathbf a )) } 
= \mathscr L_{-\mathbf a } \mathcal P \mathcal T \mathbf u ( z ) ,\]
giving the last equality in \eqref{eq:H2PT}.
\end{proof}
While all the previous symmetries applied to the full continuum model and, therefore, in particular also to the chiral and anti-chiral models, the following symmetry only applies when $\varphi=0$.

\begin{lemm}[Particle-hole symmetry]
\label{l:chir}
Define
\[ (\mathscr S \mathbf v)(z):= \begin{pmatrix} \sigma_2 & 0 \\
0  & \sigma_2 \end{pmatrix} \mathbf v(-z), \ \  \ \sigma_2 := 
\begin{pmatrix} 0 & - i \\
i & \ 0 \end{pmatrix} .  \] Then 
\begin{equation}
\label{eq:S2H}  \mathscr S H(w,0) \mathscr S = -H(w,0), \ \ \ 
\mathscr S^*=\mathscr S^{-1}=\mathscr S ,\end{equation}
and
\begin{equation}
\label{eq:S2G}
\mathscr C \mathscr S = 
\mathscr S \mathscr C , \ \ \ \mathscr L_{\mathbf a } \mathscr S = \omega^{a_1 + a_2}  \mathscr S \mathscr L_{-\mathbf a }, \ \ \mathbf a := \tfrac 43
\pi i ( a_1 \omega + a_2 \omega^2 ). 
\end{equation}
\end{lemm}
\begin{proof}
To see this notice that 
\begin{equation}
\begin{split}
\begin{pmatrix} \sigma_2 & 0 \\
0  & \sigma_2 \end{pmatrix} H(w,0) \begin{pmatrix} \sigma_2 & 0 \\
0  & \sigma_2 \end{pmatrix} = \begin{pmatrix} \sigma_2 C ( w_0) 
\sigma_2  & (\sigma_2 D(w_1,0) \sigma_2 )^*  \\  \sigma_2 D(w_1,0) \sigma_2  &  \sigma_2 C ( w_0) 
\sigma_2   \end{pmatrix}, 
\end{split}
\end{equation}
and 
\[   \sigma_2 C ( w_0 ) \sigma_2 = - w_0 \begin{pmatrix}
0 & V ( - z ) \\
V ( z ) & 0 \end{pmatrix} , \ \ 
\sigma_2 D ( w_1,0 ) \sigma_2 =
- \begin{pmatrix} 2 D_{-\bar z } &  w_1U ( - z ) \\
w_1 U ( z ) & 2 D_{ - \bar z } \end{pmatrix}, \]
which shows \eqref{eq:S2H}. To see \eqref{eq:S2G},  we consider
the action on $ \CC^2 $ in the block decomposition of $ \CC^4 $:
if $ L_{\mathbf a } u := u ( z + \mathbf a ) $ and 
$ R u := u ( - z ) $ then, for $ \mathbf a = \frac 43 \pi i \omega^\ell$, 
$ \ell = 1,2$, 
\[  \begin{split}  \mathscr L_{\mathbf a} \mathscr S &= 
L_{\mathbf  a} \begin{pmatrix} \omega & 0 \\ 0 & 1 \end{pmatrix} 
\begin{pmatrix} 0 & - i \\ i & \ 0 \end{pmatrix} R =
\begin{pmatrix} 0 & - i\omega  \\ i & \ 0 \end{pmatrix} R L_{-\mathbf a }
= R \begin{pmatrix} 0 & - i\omega  \\ i & \ 0 \end{pmatrix} 
\begin{pmatrix} \omega & 0 \\ 0 & 1 \end{pmatrix} \mathscr L_{-\mathbf a} \\
& = \begin{pmatrix} 0 & - i\omega  \\ i\omega  & \ 0 \end{pmatrix} R \mathscr L_{-\mathbf a}
= \omega \begin{pmatrix} 0 & - i \\ i & \ 0 \end{pmatrix} R \mathscr L_{-\mathbf a } = \omega \mathscr S \mathscr L_{-\mathbf a } .
\end{split}
\]
Since $ \mathscr S \mathscr C = \mathscr C \mathscr S $ is easy to check, this completes the proof.
\end{proof}

\section{Symmetries protect the kernel of \boldmath $ H ( w,0 ) $}

We recall the notation from \cite[\S 2.2]{BEWZ}:
for $  ( k, p ) \in \ZZ_3^2 $
\begin{equation}
\label{eq:1drep}
L^2_{ k, p} ( \CC/ \Gamma; \CC^4 ) := \{ 
\mathbf u \in L^2 ( \CC/ \Gamma; \CC^4 ) : \mathscr L_{\mathbf a} 
\mathbf u = \omega^{ k ( a_1 + a_2 ) } \mathbf u , \ \
\mathscr C \mathbf u = \bar \omega^{ p } \mathbf u  \} ,  
\end{equation}
where $ \mathbf a = \frac43 \pi i ( a_1 \omega + a_2 \omega^2 ) $.
Because $ H ( w ,0) $ commutes with the $ G_3$ action, we obtain
(unbounded) operators 
\[  H_{k,p} ( w,0 ) : L^2_{ k, p} ( \CC/ \Gamma; \CC^4 ) \to
L^2_{ k, p} ( \CC/ \Gamma; \CC^4 ) . \]

We now look at the action of the particle-hole and mirror symmetries on 
these spaces (Lemmas \ref{l:mir} and \ref{l:chir} respectively):
\begin{equation}
\label{eq:sym2rep}
\mathscr S : L^2_{ k,p} \to L^2_{ -k+1, p} , \ \ \ \
\mathscr M : L^2_{ k, p} \to L^2_{ -k+1, -p+1} .
\end{equation}
\begin{proof}[Proof of \eqref{eq:sym2rep}]
We see that for $ \mathbf u \in L^2_{k,p} $, 
\[ \mathscr L_{\mathbf a} \mathscr S \mathbf u = 
\mathscr S \omega^{a_1 + a_2 } \mathscr L_{ - \mathbf a } \mathbf u =
\omega^{ (-k + 1 ) ( a_1 + a_2 ) } \mathscr S \mathbf u . \]
Since $ \mathscr C $ commutes with $ \mathscr S $, this proves the first
claim.
Since $ F ( \bar{\mathbf a } ) = (-a_2, -a_1 ) $, the same argument applies to $ \mathscr M $. We also have
\[ \mathscr C \mathscr M \mathbf u =  \bar \omega \mathscr M \mathscr C^{-1}
\mathbf u = \bar \omega^{-p+1} \mathscr M \mathbf u ,\]
which completes the proof.
\end{proof}
From \eqref{eq:sym2rep} and the commutation relations
$ \mathscr  S H ( w,0 ) = - H( w,0 ) \mathscr S $, $ 
\mathscr M H ( w,0 ) = H ( w,0 ) \mathscr M $, we obtain for 
$ w \in \RR^2 $, 
\begin{equation}
\label{eq:symH1}  
\begin{split} \Spec ( H_{k,p} ( w,0 ) ) & = - \Spec ( H_{-k+1, p}  ( w,0 ) ), \\  \Spec ( H_{k, p} ( w ,0) ) & = \Spec ( H_{ -k + 1 , - p +1 } ( w ,0) ) . \end{split}  \end{equation}

We now examine the action of the $ \mathcal P \mathcal T $ symmetry:
\begin{equation}
\label{eq:PT2rep}
\mathcal P \mathcal T : L^2_{ k, p } \to L^2_{k,-p+1} . 
\end{equation}
In fact,
\[  \mathscr L_{\mathbf a } \mathcal P\mathcal T \mathbf u = 
\mathcal P \mathcal T \mathscr L_{ - \mathbf a } \mathbf u= 
\mathcal P \mathcal T \omega^{-(a_1+a_2)} \mathbf u = 
\omega^{a_1 + a_2} \mathcal P \mathcal T \mathbf u, \]
and 
\[ \mathscr C \mathcal P\mathcal T \mathbf  u  =
\bar \omega \mathcal P \mathcal T \mathscr C \mathbf u = 
 \bar \omega  
\mathcal P \mathcal T \bar \omega^p \mathbf u = 
\bar \omega^{ - p + 1} \mathcal P \mathcal T \mathbf u . \]

Hence, in addition to \eqref{eq:symH1} we also have
\begin{equation}
\label{eq:symH2}  \Spec ( H_{ k, p }  ( w,0)) = \Spec ( H_{ k, -p + 1} ( w,0) ). 
\end{equation}

Using successively \eqref{eq:symH1} and \eqref{eq:symH2}, we conclude that
\begin{equation}
\label{eq:even}   \begin{split} \Spec ( H_{k,p} ( w,0) ) & = - \Spec ( H_{ -k+1, p } ( w ,0) )
= - \Spec ( H_{k , -p+1 } ( w,0 ) )\\
&  = - \Spec ( H_{k, p } ( w,0 )). \end{split}
\end{equation}
Hence the spectrum on $ L^2_{k,p} $ is invariant under 
$ \lambda \mapsto - \lambda $ and we have, as in \cite[\S 2.2]{BEWZ},
\begin{equation}
\label{eq:protected}
\ker_{ L^2_{ k,p} ( \CC/\Gamma; \CC^4 ) } H ( w,0 ) \neq \{ 0 \}, \ \ \
w \in \RR^2 , \ \   k, p \in \{ 0, 1 \} . 
\end{equation}

\section{Absence of flat bands in anti-chiral model}

We now show that there are no flat bands in the anti-chiral model, $\varphi=w_1=0$, at zero energy.

\begin{theo}
\label{theo:noflatbands}
For the anti-chiral model $w_0 \neq 0$ and $\varphi = w_1=0$, there are no flat bands at energy zero.
\end{theo}
\begin{proof}
The preliminary discussions shows that we are interested in the invertibility of the self-adjoint operator
\[ \mathscr H_{\mathbf k}(w_0):= \begin{pmatrix} 0 & Q_{\mathbf k}(w_0) \\
Q_{\mathbf k}(w_0)^* & 0 \end{pmatrix},\] where $Q_{\mathbf k}(w_0)$ is the operator \eqref{eq:Q2} after applying the Floquet transform with quasi-momentum $\mathbf k=\mathbf k_1 + i \mathbf k_2 $ and $\mathbf k_i \in \RR.$
The operator $\RR \ni \mathbf k_1 \mapsto \mathscr H_{\mathbf k}$  for $\mathbf k_2$ fixed is a self-adjoint holomorphic family with compact resolvent on $L^2(\RR^2/\Gamma)$. 
We then assume that $\mathscr H_{\mathbf k}(w_0)$ is not invertible for any $\mathbf k \in \CC $, i.e., $\mathscr H_{\mathbf k}(w_0)$ has a flat band at energy zero.

To simplify the analysis, we express $ Q_{\mathbf k}(w_0)$ and its adjoint in terms of Pauli matrices:
\begin{equation}
\begin{split}
Q_{\mathbf k}(w_0) &= \underbrace{\sum_{j=1}^2 (D_j -\mathbf k_j)\sigma_j}_{=:H_{\text{Dirac}}(\mathbf k)} + iw_0  \Im(V(z,\bar z))\sigma_3 + w_0\Re( V(z,\bar z))\operatorname{id} \\
Q_{\mathbf k}(w_0)^* &= \underbrace{\sum_{j=1}^2 (D_j -\mathbf k_j)\sigma_j}_{=:\hat{H}_{\text{Dirac}}(\mathbf k)} - iw_0 \Im( V(z,\bar z))\sigma_3 +w_0  \Re(V(z,\bar z))\operatorname{id}.
\end{split}
\end{equation}

In this setting, we have that both $\mathbf k_1,\mathbf k_2$ are real. We now complexify the real part of $\mathbf k$, which is $\mathbf k_1$, and choose $\mathbf k= \mathbf k_1+ i \mathbf k_2 $ with $\mathbf k_1:=(\mu+ i y)  $, where $\mu, y , \mathbf k_2 \in \RR.$

We thus have that 
\begin{equation}
\begin{split}
 (Q_{\mathbf k}(w_0))^2 &= H_{\text{S}}(\mathbf k) + Q_{\mathbf k}(w_0) w_0 \Re(V(z)) + W \\
  (Q_{\mathbf k}(w_0)^*)^2 &= H_{\text{S}}(\mathbf k) + Q_{\mathbf k}(w_0)^*w_0 \Re(V(z)) + \hat{W}
 \end{split}
\end{equation}
where $W,\hat{W} \in L^{\infty}(\CC;\CC^{2\times 2})$ are given by
\begin{equation}
\begin{split}
W := &-\vert \alpha_0 V(z) \vert^2 \operatorname{id}+ i e^{i  \varphi/2}\left(\partial_2 \Im (\alpha_0V(z))\sigma_1-\partial_1 \Im (\alpha_0V(z)) \sigma_2 \right)\\
&\quad + i (\partial_1\Re (\alpha_0V(z)) \sigma_1+  \partial_2\Re (\alpha_0V(z)) \sigma_2)\\
\hat{W}:=&-\vert \alpha_0 V(z) \vert^2 \operatorname{id}- i e^{-i  \varphi/2}\left(\partial_2 \Im (\alpha_0V(z))\sigma_1-\partial_1 \Im (\alpha_0V(z)) \sigma_2 \right)\\
&\quad + i (\partial_1\Re (\alpha_0V(z)) \sigma_1+  \partial_2\Re (\alpha_0V(z)) \sigma_2)
\end{split}
\end{equation}
and $H_{\text{S}}(\mathbf k)$ is the Schr\"odinger operator
\[  H_{\text{S}}(\mathbf k)=H_{\text{Dirac}}(\mathbf k)^2 = \left((D_1+\mu + i y)^2+(D_2+k_2)^2 \right) \operatorname{id}\] 
satisfying 
\begin{equation}
\label{eq:Pauli}
\Vert H_{\text{S}}(\mathbf k)^{-1} \Vert \lesssim \langle y \rangle^{-2}.
\end{equation}

Rearranging terms yields
\begin{equation}
\begin{split}
 &Q_{\mathbf k}(w_0)(Q_{\mathbf k}(w_0) - w_0  \Re(V(z,\bar z)) ) =  (\operatorname{id}+WH_{\text{S}}(\mathbf k)^{-1})H_{\text{S}}(\mathbf k) \\
 &Q_{\mathbf k}(w_0)^*(Q_{\mathbf k}(w_0)^*- w_0\Re( V(z,\bar z)) )=  (\operatorname{id}+\hat{W}H_{\text{S}}(\mathbf k)^{-1})H_{\text{S}}(\mathbf k).
 \end{split}
 \end{equation}

This shows, since the right hand side of \eqref{eq:identity} is invertible, that $Q_{\mathbf k}(w_0)$ is surjective for $y$ large enough. Applying the same argument to $Q_{\mathbf k}(w_0)^*$ and using that $\operatorname{ker}(Q(w_0)) = \operatorname{ran}(Q_{\mathbf k}(w_0)^*)^{\perp}$ yields invertibility of $Q_{\mathbf k}(w_0)$ and thus of $Q_{\mathbf k}(w_0)^*.$

In particular, we can easily write down the inverse of $\mathscr H_{\mathbf k}(w_0)$ as
\[  \mathscr H_{\mathbf k}(w_0)^{-1}= \begin{pmatrix} 0 & (Q_{\mathbf k}(w_0))^{-1} \\ (Q_{\mathbf k}(w_0)^*)^{-1} & 0 \end{pmatrix} . \]
Using operators
\begin{equation}
\begin{split}
\hat{T}&:=(Q_{\mathbf k}(w_0)^* - w_0  \Re(V(z))) H_{\text{S}}(\mathbf k)^{-1} \\
T&:=(Q_{\mathbf k}(w_0) - w_0 \Re(V(z)))H_{\text{S}}(\mathbf k)^{-1} 
\end{split}
\end{equation}
we have
 \begin{equation}
\begin{split}
 \mathscr H_{\mathbf k}(w_0)^{-1} 
 &=\begin{pmatrix} 0 & Q_{\mathbf k}(w_0)^{-1} \\ (Q_{\mathbf k}(w_0)^{*})^{-1} & 0 \end{pmatrix}  \\
  &=\begin{pmatrix} 0 & T \\ \hat{T}  & 0 \end{pmatrix} \left(\operatorname{id} +\begin{pmatrix} \hat{W}H_{\text{S}}(\mathbf k)^{-1} & 0 \\ 0 &  W H_{\text{S}}(\mathbf k)^{-1} \end{pmatrix} \right)^{-1}.
  \end{split}
 \end{equation}
Assuming that there exists a flat band to $\mathscr H_{\mathbf k}$, it follows that 
\[1 \in \Spec(\operatorname{diag}(\hat{W}H_{\text{S}}(\mathbf k)^{-1} ,WH_{\text{S}}(\mathbf k)^{-1}))\] in a complex neighbourhood of $\mathbf k_1 \in \RR$ by Rellich's theorem. 
Then it follows that for all $\mathbf k_1 \in \CC$ 
 
 \[1 \in \Spec(\operatorname{diag}(\hat{W}H_{\text{S}}(\mathbf k)^{-1} ,WH_{\text{S}}(\mathbf k)^{-1})).\]
 But this contradicts the resolvent estimate on $H_{\text{S}}(\mathbf k)^{-1} $ for $\vert y \vert$ large enough.
\end{proof}

\section{Point-localized states and H\"ormander's condition}
\smallsection{Notation}
For an open set $\Omega\subset\mathbb R^n$ we denote by $T^\ast\Omega$ the $2n$-dimensional phase space over $\Omega$. For a function $a$ of phase space coordinates $x,\xi$ satisfying certain regularity we denote by $a^w(x,hD_{x})$ or $a^{w}$ the semiclassical Weyl quantization of $a$, see \cite[Def.\@ $4.1$]{zworski}. We denote the semiclassical wavefront set by $\operatorname{WF}_h$, see \cite[Sec.\@ $8.4$]{zworski} for details. The principal symbol of a pseudodifferential operator $A$ is denoted by $\sigma_0(A)$, see \cite[p.\@ $213$]{zworski}.

\bigskip

\subsection{Absence of point-localized states in anti-chiral model}

We shall now establish the absence of point-localized states in the anti-chiral model. In view of \eqref{eq:Q}--\eqref{eq:Q2} is suffices to study quasimodes for $Q_\mathbf{k}(w_0)$, which as before is the operator \eqref{eq:Q2} after applying the Floquet transform. It is then natural to consider $h=1/w_0$ as a semiclassical parameter, so that
\begin{equation*}
Q_\mathbf{k}(w_0)=h^{-1}\widetilde Q_\mathbf{k}(h),
\quad 
\widetilde Q_\mathbf{k}(h):=\begin{pmatrix}
V(z, \bar z)  & 2 hD_z-h\bar{\mathbf k}  \\
2 hD_{\bar z}-h\mathbf k& \overline{V(z, \bar z)}
\end{pmatrix} .
\end{equation*}
Introduce the set 
\[\mathcal A:=\left\{ \tfrac{4\pi(3m\pm1)}{3 \sqrt{3}}+i \tfrac{2\pi n}{3}; n \in 2\ZZ, m \in \ZZ\right\} \cup \left\{ \tfrac{2\pi(6m\pm1)}{3 \sqrt{3}}+i\tfrac{2\pi n}{3}; n \in 2\ZZ+1, m \in \ZZ\right\}\]
and observe that
$$
V(z,\bar z)
=e^{i\Im z}\bigg(1+2\cos\bigg(\frac{\sqrt 3}2 \Re z\bigg)e^{-\frac32 i\Im z}\bigg)=0\quad\Longleftrightarrow\quad z\in\mathcal A.
$$

Let $\Omega\subset \CC\setminus\mathcal A$ be open and uniformly bounded away from $\mathcal A$. We shall use complex notation $\zeta=\frac12(\xi_1-i\xi_2)$ and $z=x_1+ix_2$ in $T^\ast \Omega$.
Let $z_0\in\Omega$ and suppose that $\mathbf u_h=(u_h,v_h)\in C^\infty(\CC/\Gamma;\CC^2)$ is a quasimode of $\widetilde Q_\mathbf{k}(h)$ with $\operatorname{WF}_h(u_h)=\operatorname{WF}_h(v_h)=\{(z_0,\zeta_0)\}$ such that $\widetilde Q_\mathbf{k}(h)\mathbf u_h=O(h^\infty)$ as $h\to0$. Then also $Q_\mathbf{k}(w_0)\mathbf u_h=O(h^\infty)$ as $w_0=1/h\to\infty$. Moreover, it is easy to see that $P(h)v_h=O(h^\infty)$ for the scalar operator
\begin{align*}
P(h)v_h(z)=V(z,\bar z)(2hD_{\bar z}-h\mathbf k)\big[V(z,\bar z)^{-1}(2hD_z-h\bar{\mathbf{k}}) v_h(z)\big]-\lvert V(z,\bar z)\rvert^2v_h(z);
\end{align*}
compare with the proof of Lemma \ref{lemm:equivalence} below. Since $\operatorname{WF}_h(v_h)= \{(z_0,\zeta_0)\}$, this means that a quasimode $\mathbf u_h$ of $Q_\mathbf k(w_0)$, should one exist, would give rise to a quasimode $v_h$ of $P(h)$ microlocalized in $T^\ast\Omega$. This is not possible, by virtue of the following result.

%where we use the terminology {\it localized in phase space} for a tempered family of functions $u(h)\in C^\infty(\CC/\Gamma;\CC^2)$ if 
%\begin{equation}\label{localizationproperty}
%\diag(\operatorname{id}-\chi^w)u(h)=O_\mathscr{S}(h^\infty)
%\end{equation}
%for some compactly supported scalar symbol $\chi$ with $\supp(\chi)\subset T^\ast\Omega$. (See \cite[Sec.\@ $8.4.1$]{zworski} for details.)

\begin{theo}
\label{theo:Jens}
If $u(h)\in C^\infty(\CC/\Gamma;\CC^2)$ 
%is localized in phase space 
with $\operatorname{WF}_h(u(h))=\{(z_0,\zeta_0)\}$ contained in $T^\ast\Omega$, then
\begin{equation}\label{eq:PTestimate}
\lVert u(h)\rVert\le\frac{C}{h}\lVert P(h) u(h)\rVert,\quad h\to 0.
\end{equation}
In particular, there are no quasimodes of the anti-chiral model $w_0\ne0$, $\varphi=w_1=0$, microlocalized at $(z_0,\zeta_0)\in T^\ast\Omega$.
\end{theo}

\begin{proof}
In view of the previous discussion, we only need to verify \eqref{eq:PTestimate}.
Let $P(z,\bar z,\zeta,\bar \zeta)$ be the symbol of $P(h)=P^w(z,\bar z,hD_z,hD_{\bar z})$. It is easy to see that $P$ has the regularity of belonging to the symbol class $S(m)$, where $m(z,\bar z,\zeta,\bar \zeta)=1+\lvert z\rvert^2+\lvert\zeta\rvert^2$. (See \cite[Ch.\@ 4]{zworski} for details.) The principal symbol of $P(h)$ is 
$$
p(z,\bar z,\zeta,\bar\zeta)
%=V(z,\bar z)(2\bar\zeta)\big[V(z,\bar z)^{-1}(2\zeta)]-\lvert  V(z,\bar z)\rvert^2
=\lvert 2\zeta\rvert^2-\lvert  V(z,\bar z)\rvert^2.
$$
We may assume that $p(z_0,\bar z_0,\zeta_0,\bar \zeta_0)=0$, since otherwise the result follows by ellipticity. Note that
$$
p(z,\bar z,\zeta,\bar\zeta)=\lvert 2\zeta\rvert^2-\lvert V(z,\bar z)\rvert^2=0\quad\Longrightarrow\quad \zeta\ne0 
$$
in $T^\ast\Omega$, so $P(h)$ is of real principal type. In particular, the Hamilton vector field $H_p$ is non-vanishing and not colinear with the radial vector field $\xi\cdot\partial_\xi$ at $(z_0,\zeta_0)$. 
Thus, \eqref{eq:PTestimate} follows from \cite[Theorem 12.4]{zworski}. (Actually, \cite[Theorem 12.4]{zworski} concerns operators on $\mathbb R^n$, but a similar result holds for operators on $\CC/\Gamma$. This can be verified by introducing rectangular coordinates $ z = x_1 + i x_2 =  2i \omega y_1 + 2 i \omega^2 y_2 $, and identifying $P(h)$ with the resulting operator in $(y_1,y_2)$ belonging to the two-dimensional torus; see, e.g.,~\cite[Ch.~5.3]{zworski}.) % and \cite{RT}.) 
This completes the proof.
\end{proof}

\subsection{The general continuum model}
We shall now look into a method that applies to the general continuum model, which is given by a Hamiltonian that is not of block off-diagonal form.
In particular, as we shall show, H\"ormander's condition does not apply to the general continuum model, and the bracket almost always vanishes. The purpose of this subsection is therefore to outline a general reduction strategy that reduces the $4\times 4$ Hamiltonian to a $2\times 2$ system, and then to show that the Poisson bracket of the determinant of that system always vanishes.

The following Lemma connects the existence of quasimodes to $2n \times 2n$-matrix valued systems to the existence of quasimodes to $n \times n$ matrix-valued systems.
\begin{lemm}\label{lemm:equivalence}
Let $A$ be a self-adjoint matrix-valued symbol $A \in \RR^{2n \times 2n}$ for some $n \in \mathbb N$ 
\begin{equation}
A(x,\xi)= \begin{pmatrix} A_{11}(x,\xi) & A_{12}(x,\xi) \\ A_{12}(x,\xi)^*&  A_{22}(x,\xi),
\end{pmatrix}
\end{equation}
where $A_{12}(x,\xi)$ is invertible away from a set $\mathcal A.$ 
Thus, if for some $\varphi= (\varphi_1, \varphi_2)^t$ we have $\operatorname{WF}_h(\varphi)=\{(x_0,\xi_0)\} \notin \mathcal A$, then 
\[A^w(x,hD_x)\varphi=\mathcal O(h^{\infty})\text{ is equivalent to }F^w(x,hD_x) \varphi_1 =\mathcal O(h^{\infty})\]
where 
\begin{equation}
\label{eq:F}
F^w(x,hD_x) :=A_{12}^w(x,hD_x)^*-A_{22}^w(x,hD_x)
A_{12}^w(x,hD_x)^{-1} A_{11}^w(x,hD_x).
\end{equation}
\end{lemm}
\begin{proof}
Let $\varphi = (\varphi_1,\varphi_2)^t$ with $\operatorname{WF}_h(\varphi)=\{(x_0,\xi_0)\}$  be a quasimode for some $(x_0,\xi_0) \notin \mathcal A$ satisfying 
\[ A^w(x,hD_x)\varphi=\mathcal O(h^{\infty}), \]
then \begin{equation}
\begin{split}
\label{eq:matrix-comp}
&A^w_{11}(x,hD_x)\varphi_1 + A^w_{12}(x,hD_x) \varphi_2 = \mathcal O(h^{\infty})\text{ and }\\
&A^w_{12}(x,hD_x)^*\varphi_1+  A^w_{22}(x,hD_x) \varphi_2 =\mathcal O(h^{\infty}).
\end{split}
\end{equation}
By the microlocal property of pseudodifferential operators \cite[Thm.~$8.15$]{zworski} it then follows that $\operatorname{WF}_h(A_{11}^w(x,hD_x)\varphi_1) \subset \{(x_0,\xi_0)\}.$ If $(x_0,\xi_0) \notin \mathcal A$, then $\operatorname{det}(A^w_{12}(x_0,\xi_0)) \neq 0$ and we can microlocally invert $A^w_{12}(x,hD_x)^{-1}$ such that 
\[ \varphi_2=-A_{12}^w(x,hD_x)^{-1} A_{11}^w(x,hD_x)\varphi_1 +\mathcal O(h^{\infty}).\]
Inserting this into the second equation in \eqref{eq:matrix-comp} we find that 

\[F^w(x,hD_x) \varphi_1 =\mathcal O(h^{\infty}).\]

\end{proof}

The above reduction from $A$ to the lower-dimensional symbol $F$ has an interesting consequence for the Poisson bracket of non-degenerate eigenvalues of $F$, if $A$ has real determinant. In fact, it shows that we cannot use the non-self-adjointness of $F$ to construct quasimodes to $F$ (and thus to $A^w(x,hD_x)$ by the previous Lemma) using H\"ormander's condition applied to simple zero eigenvalues of $\sigma_0(F).$ In particular, the following corollary therefore implies that the Poisson bracket will vanish at most critical points of the generalized continuum model.

\begin{corr}
\label{corr:simple}
Let $(x,\xi)$ be a point such that all eigenvalues $\lambda_1,\ldots,\lambda_n$ of $\sigma_0(F(x,\xi))$ are smooth and  $\det(A)$ is real  in a neighbourhood of $(x,\xi)$. If all eigenvalues apart from $\lambda_1$ are non-zero at $(x,\xi)$, then
\[ \left\{ \lambda_1, \overline{\lambda_1} \right\}(x,\xi)=0.\]
\end{corr}
\begin{proof}

For the principal symbol of $F$ we then have for $(x,\xi)$, where $\operatorname{det}(A^{}(x,\xi)) \neq 0,$
\begin{equation}
\begin{split}
\label{eq:blockform}
\operatorname{det}(A^{}(x,\xi))&=\operatorname{det}( A_{12}(x,\xi)) \operatorname{det}( \sigma_0(F(x,\xi))).
\end{split}
\end{equation}

Let $(x,\xi)$ be a point such that all eigenvalues $\lambda_1,\ldots,\lambda_n$ of $\sigma_0(F(x,\xi))$ are smooth in a neighbourhood of $(x,\xi)$ and all but $\lambda_1$ non-zero at $(x,\xi)$. Then 
\begin{equation}
\begin{split}
\left\{ \lambda_1, \overline{\lambda_1} \right\}(x,\xi) &= \frac{\left\{\operatorname{det}(\sigma_0(F)), \overline{\operatorname{det}(\sigma_0(F))} \right\}(x,\xi)}{\prod_{i=2}^n \vert \lambda_i(x,\xi)\vert^2} \\[.5em]
&= \frac{\left\{\operatorname{det}(A),\operatorname{det}(A) \right\}(x,\xi)}{\vert \operatorname{det}(A_{12}(x,\xi)) \vert^2 \prod_{i=2}^n \vert \lambda_i(x,\xi)\vert^2 }=0.
\end{split}
\end{equation}
\end{proof}

\end{document}